\documentclass[runningheads,svgnames]{llncs}
\usepackage[T1]{fontenc}
\usepackage{graphicx}

\usepackage{amsmath,amsfonts,amssymb}
\usepackage{tikz,tikz-3dplot}
\usepackage{bm}
\usetikzlibrary{shapes,snakes}
\usepackage{bookmark}
\begin{document}
	\title{Evolomino is NP-complete}
	%
	%
	\author{Andrei V. Nikolaev\orcidID{0000-0003-4705-2409}}
	\authorrunning{A.V. Nikolaev}
	%
	\institute{P.G. Demidov Yaroslavl State University, Yaroslavl, Russia
		\email{andrei.v.nikolaev@gmail.com}}

	\maketitle              
	\begin{abstract}
		Evolomino is a pencil-and-paper logic puzzle popularized by the Japanese publisher Nikoli (like Sudoku, Kakuro, Slitherlink, Masyu, and Fillomino). The puzzle's name reflects its core mechanic: the shapes of polyomino-like blocks that players must draw gradually ``evolve'' in the directions indicated by pre-drawn arrows. We prove, by reduction from \textsc{3-SAT}, that the question of whether there exists at least one solution to an Evolomino puzzle satisfying the rules is NP-complete. Since our reduction is parsimonious, i.e., it preserves the number of distinct solutions, we also prove that counting the number of solutions to an Evolomino puzzle is $\texttt{\#}$P-complete.
		
		\keywords{Computational complexity  \and NP-complete \and $\texttt{\#}$P-complete \and Evolomino \and pencil-and-paper logic puzzle \and polynomial-time reduction \and \textsc{3-SAT} \and parsimonious reduction.}
	\end{abstract}
	\section{Introduction}
	
	Nikoli is a renowned Japanese publisher specializing in games and, especially, logic puzzles. Established in 1980, it became prominent worldwide with the popularity of Sudoku.
	The magazine published under the same name by Nikoli was the first puzzle magazine in Japan, and over the years, many new types of puzzles appeared on its pages.
	
	In this paper, we consider a new pencil-and-paper logic puzzle Evolomino that was introduced in the book \textit{768 The Pencil Puzzles 2025} published by Nikoli in 2025~\cite{NikoliPuzzleBook}.
	
	Here we present the rules of the Evolomino puzzle as they appear on Nicoli's official website~\cite{Evolomino_Nikoli_web_page}:
	
	\begin{itemize}
		\item Evolomino is played on a rectangular board with white and shaded cells. Some white cells contain pre-drawn squares and arrows.
		
		\item The player must draw squares ($\square$) in some of the white cells.
		
		\item A polyomino-like group of squares connected vertically and horizontally is called a \textit{block} (including only one square). Each block must contain exactly one square placed on a pre-drawn arrow.
		
		\item Each arrow must pass through at least two blocks.
		
		\item The second and later blocks on the route of an arrow from start to finish must progress by adding one square to the previous block without rotating or flipping.
	\end{itemize}
	
	An example of the Evolomino puzzle and its solution from Nicoli's official website~\cite{Evolomino_Nikoli_web_page} is shown in Fig.~\ref{Fig_Evolomino_example}.
	
	\begin{figure}[t]
		\centering
		\begin{tikzpicture}[scale=0.65]
			
			\foreach \x in {0,...,5}
			\draw (\x+0.5,0.5)--(\x+0.5,5.5);
			
			\foreach \y in {0,...,5}
			\draw (0.5,\y+0.5)--(5.5,\y+0.5);

			\foreach \x in {5}	
			\foreach \y in {1}
			\draw [fill = MidnightBlue!30] (\x-0.5,\y-0.5) -- (\x+0.5,\y-0.5) -- (\x+0.5,\y+0.5) -- (\x-0.5,\y+0.5) -- cycle;	
			
			\foreach \x in {4}	
			\foreach \y in {4}
			\draw [fill = MidnightBlue!30] (\x-0.5,\y-0.5) -- (\x+0.5,\y-0.5) -- (\x+0.5,\y+0.5) -- (\x-0.5,\y+0.5) -- cycle;

			\draw [semithick,->,>=stealth] (1,1) -- (3+0.1,1);
			\draw [semithick,->,>=stealth] (1,3) -- (1,5) -- (4+0.1,5);
			\draw [semithick,->,>=stealth] (5,2) -- (4,2) -- (4,3+0.1);
			
			\foreach \x in {1}
			\foreach \y in {3}
			\draw (\x-0.25,\y-0.25) -- (\x+0.25,\y-0.25) -- (\x+0.25,\y+0.25) -- (\x-0.25,\y+0.25) -- cycle;

			\foreach \x in {5}
			\foreach \y in {4}
			\draw (\x-0.25,\y-0.25) -- (\x+0.25,\y-0.25) -- (\x+0.25,\y+0.25) -- (\x-0.25,\y+0.25) -- cycle;
			
			\node at (3,0) {Sample puzzle};
			
			\begin{scope}[xshift=8cm]
				
				\foreach \x in {0,...,5}
				\draw (\x+0.5,0.5)--(\x+0.5,5.5);
				
				\foreach \y in {0,...,5}
				\draw (0.5,\y+0.5)--(5.5,\y+0.5);

				\foreach \x in {5}	
				\foreach \y in {1}
				\draw [fill = MidnightBlue!30] (\x-0.5,\y-0.5) -- (\x+0.5,\y-0.5) -- (\x+0.5,\y+0.5) -- (\x-0.5,\y+0.5) -- cycle;	
				
				\foreach \x in {4}	
				\foreach \y in {4}
				\draw [fill = MidnightBlue!30] (\x-0.5,\y-0.5) -- (\x+0.5,\y-0.5) -- (\x+0.5,\y+0.5) -- (\x-0.5,\y+0.5) -- cycle;

				\draw [semithick,->,>=stealth] (1,1) -- (3+0.1,1);
				\draw [semithick,->,>=stealth] (1,3) -- (1,5) -- (4+0.1,5);
				\draw [semithick,->,>=stealth] (5,2) -- (4,2) -- (4,3+0.1);
				
				\foreach \x in {1}
				\foreach \y in {1,3}
				\draw (\x-0.25,\y-0.25) -- (\x+0.25,\y-0.25) -- (\x+0.25,\y+0.25) -- (\x-0.25,\y+0.25) -- cycle;
				
				\foreach \x in {2}
				\foreach \y in {4,5}
				\draw (\x-0.25,\y-0.25) -- (\x+0.25,\y-0.25) -- (\x+0.25,\y+0.25) -- (\x-0.25,\y+0.25) -- cycle;
				
				\foreach \x in {3}
				\foreach \y in {1,3}
				\draw (\x-0.25,\y-0.25) -- (\x+0.25,\y-0.25) -- (\x+0.25,\y+0.25) -- (\x-0.25,\y+0.25) -- cycle;
				
				\foreach \x in {4}
				\foreach \y in {1,3,5}
				\draw (\x-0.25,\y-0.25) -- (\x+0.25,\y-0.25) -- (\x+0.25,\y+0.25) -- (\x-0.25,\y+0.25) -- cycle;
				
				\foreach \x in {5}
				\foreach \y in {2,4,5}
				\draw (\x-0.25,\y-0.25) -- (\x+0.25,\y-0.25) -- (\x+0.25,\y+0.25) -- (\x-0.25,\y+0.25) -- cycle;
				
				\node at (3,0) {Solution};
			\end{scope}
		\end{tikzpicture}
		\caption {Example of an \textsc{Evolomino} puzzle}
		\label {Fig_Evolomino_example}
	\end{figure}

	The study of the computational complexity of games and puzzles is a rapidly growing area within theoretical computer science.
	This interest is driven by several factors. 
	On the one hand, from an academic and educational perspective, games and puzzles offer engaging and accessible examples of combinatorial problems, making them valuable tools for teaching computational complexity theory and proof techniques.
	On the other hand, understanding the complexity of a problem aids in designing efficient algorithms for solving puzzles and in developing winning strategies for combinatorial games involving two or more players.
	
	Many classic games are known to be computationally intractable (assuming P $\neq$ NP): one-player puzzles are often NP-complete (as Minesweeper~\cite{Kaye2000}, Plumber~\cite{Kral2004}, and Tetris~\cite{Breukelaar2004}) or PSPACE-complete (as Rush Hour~\cite{Flake2002}), and two-player games are often PSPACE-complete (as Hex~\cite{Reisch1981} and Reversi~\cite{Iwata1994}) or EXPTIME-complete (as Chess~\cite{Frankel1981} and Go~\cite{Robson1983}).
	
	In particular, many pencil-and-paper logic puzzles introduced or popularized by Nikoli are NP-complete: Sudoku~\cite{Yato2003}, Kakuro~\cite{Ruepp2010,Seta2002}, Hashiwokakero~\cite{Anderson2009}, Numberlink~\cite{Kotsuma2010}, Shikaku and Ripple effect~\cite{Yasuhiko2013}, Shakashaka~\cite{Demaine2014}, Tatamibari~\cite{Adler2020}, and many others. 
	For more results on the computational complexity of games and puzzles, see the surveys by Demaine and Hearn~\cite{Demaine2009} and Kendall et al.~\cite{Kendall2008}.
	
	In this paper, we prove that determining whether an Evolomino puzzle has at least one valid solution is NP-complete, via a reduction from \textsc{3-SAT}.
	Furthermore, because our reduction preserves the number of solutions, we also establish that counting the total number of solutions to an Evolomino puzzle is $\texttt{\#}$P-complete.

	\section{Problem statement}
	
	Let us begin with a formal definition of the Evolomino puzzle.
	
	\vspace{2mm}
	
	\textbf{\textsc{Evolomino.}}
	
	\textsc{Instance.} Given a rectangular board of size $p \times q$, where each cell is either white or shaded. Some of the white cells contain pre-drawn squares ($\square$). The board also includes a set of pre-drawn arrows. Each arrow starts at the center of one white cell and ends at the center of another, traversing through the centers of intermediate white cells horizontally, vertically or with 90 degree turns. No two arrows may pass through the same cell.
	
	\textsc{Question.} Does there exist at least one solution to the puzzle, i.e., a map from the set of white cells to the set $\{\emptyset, \square\}$ that satisfies the rules of the puzzle:
	\begin{itemize}
		\item each block contains exactly one square placed on the arrow; \item each arrow passes through at least two blocks;
		\item each subsequent block in the direction of the arrow adds one square to the previous block without rotating or flipping?
	\end{itemize}
	
	\vspace{2mm}
	
	We also consider a problem of \textsc{Counting Evolomino}, which has the same instance but asks how many distinct solutions the puzzle has.

	\section{Evolomino is in NP}
	
	\begin{lemma}
		\textsc{Evolomino} $\in$ NP. \label{Lemma_Evolomino_in_NP}
	\end{lemma}
	
	\begin{proof}
		\textsc{Evolomino} is a decision problem, so to prove that it belongs to the class NP, it is sufficient to show that verifying whether a given solution satisfies all the rules of the puzzle is performed in polynomial time.
		
		Let's consider an \textsc{Evolomino} puzzle on a rectangular board of size $p \times q$. If we guess some solution to the puzzle, then checking its correctness will require:
		
		\begin{itemize}
			\item Rule: ``Each block contains exactly one square placed on the arrow.'' 
			
			Since each board cell can belong to no more than one block, we have at most $O(pq)$ blocks. The time required to traverse a single block is linear in its size, which is bounded above by $p \times q$, the size of the entire board. Therefore, the overall time complexity for verifying this rule is $O(p^2 q^2)$.
			
			\item Rule: ``Each arrow must pass through at least two blocks''.
			
			Verifying this rule requires a single pass over each pre-drawn arrow. Since the total length of all arrows is bounded by the number of cells on the board, the overall time complexity for this check is $O(pq)$.
			
			\item Rule: ``The second and later blocks on the route of an arrow from start to finish must progress by adding one square to the previous block without rotating or flipping''.
			
			We consider a pair of consecutive blocks along the direction of an arrow, both of size at most $O(pq)$. For each such pair, we can attempt to exclude each square of the larger block one by one, then check whether the resulting shape matches the smaller block. This requires traversing both blocks and comparing their structures, resulting in a worst-case complexity of $O(p^2 q^2)$ for a pair of two consecutive blocks.
			
			Since the number of consecutive block pairs is bounded by the total length of all pre-drawn arrows, i.e., $O(pq)$, the overall complexity of verifying this rule is $O(p^3 q^3)$.
		\end{itemize}
		
		Summing the complexities of verifying all the rules, we conclude that a given solution to the puzzle can be verified in polynomial time $O(p^3 q^3)$. Note that this is a fairly rough upper bound, but to prove that \textsc{Evolomino} $\in$ NP, it is sufficient to be polynomial in the size of the board $p \times q$. \qed
	\end{proof}

	\section{Evolomino is NP-hard}
	
	\begin{theorem}
		\textsc{Evolomino} is NP-hard.
	\end{theorem}
	
	\begin{proof}
		
		To prove that \textsc{Evolomino} is NP-hard, we construct a polynomial-time reduction from NP-complete \textsc{3-SAT} problem~\cite{Garey1979,Karp1972} to \textsc{Evolomino}.
		
		\vspace{2mm}
		
		\textbf{\textsc{3-SAT.}}
		
		\textsc{Instance.} Given a Boolean formula in conjunctive normal form (CNF), i.e., a collection $C = \{C_1,\ldots,C_m\}$ of clauses on a finite set $X=\{x_1,\ldots,x_n\}$ of Boolean variables such that $|C_i| = 3$ for all $1 \leq i \leq m$.
		
		\textsc{Question.} Is there a truth assignment for variables $X$ that satisfies all the clauses in $C$, and therefore satisfies the Boolean formula?	
		
		\vspace{2mm}
		
		We model each component of \textsc{3-SAT} with some gadgets on the \textsc{Evolomino} board so that the puzzle has a solution if and only if the corresponding instance of \textsc{3-SAT} is satisfiable.

		\subsection*{Variable and wire gadget}
		
		The variable and wire gadget models the truth assignment of a Boolean variable. Its structure is shown in Fig.~\ref{Fig_variable_wire_gadget}.
		
		Since each block must contain exactly one square placed on an arrow, and each subsequent block along an arrow must progress by adding one square to the previous block without rotating or flipping, the gadget admits only two feasible solutions: two vertical blocks of 1 and 2 squares representing $x = 1$ (Fig.~\ref{Fig_variable_wire_gadget} (a)), and two vertical blocks of 3 and 4 squares representing $x = 0$ (Fig.~\ref{Fig_variable_wire_gadget} (b)).
		
		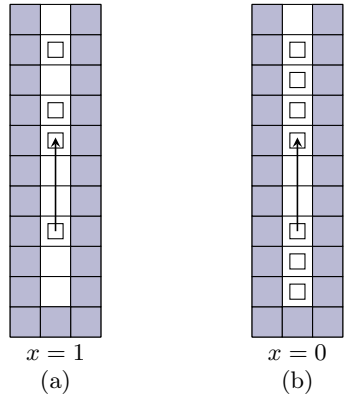
\begin{figure}[t]
			\centering
			\begin{tikzpicture}[scale=0.40]
				
				
				\node at (2,0) {$x=1$};
				
				\node at (2,-1) {(a)};
				
				\foreach \x in {0,...,3}
				\draw (\x+0.5,0.5)--(\x+0.5,11.5);
				
				\foreach \y in {0,...,11}
				\draw (0.5,\y+0.5)--(3.5,\y+0.5);
				
				\foreach \x in {1,3}	
				\foreach \y in {1,...,11}
				\draw [fill = MidnightBlue!30] (\x-0.5,\y-0.5) -- (\x+0.5,\y-0.5) -- (\x+0.5,\y+0.5) -- (\x-0.5,\y+0.5) -- cycle;	
				
				\foreach \x in {2}	
				\foreach \y in {1}
				\draw [fill = MidnightBlue!30] (\x-0.5,\y-0.5) -- (\x+0.5,\y-0.5) -- (\x+0.5,\y+0.5) -- (\x-0.5,\y+0.5) -- cycle;	
				
				\foreach \x in {2}
				\foreach \y in {4,7,8,10}
				\draw (\x-0.25,\y-0.25) -- (\x+0.25,\y-0.25) -- (\x+0.25,\y+0.25) -- (\x-0.25,\y+0.25) -- cycle;
				
				\draw [semithick,->,>=stealth] (2,4) -- (2,7+0.1);

				\begin{scope}[xshift=8cm]
					
					
					\node at (2,0) {$x=0$};
					
					\node at (2,-1) {(b)};
					
					\foreach \x in {0,...,3}
					\draw (\x+0.5,0.5)--(\x+0.5,11.5);
					
					\foreach \y in {0,...,11}
					\draw (0.5,\y+0.5)--(3.5,\y+0.5);
					
					\foreach \x in {1,3}	
					\foreach \y in {1,...,11}
					\draw [fill = MidnightBlue!30] (\x-0.5,\y-0.5) -- (\x+0.5,\y-0.5) -- (\x+0.5,\y+0.5) -- (\x-0.5,\y+0.5) -- cycle;	
					
					\foreach \x in {2}	
					\foreach \y in {1}
					\draw [fill = MidnightBlue!30] (\x-0.5,\y-0.5) -- (\x+0.5,\y-0.5) -- (\x+0.5,\y+0.5) -- (\x-0.5,\y+0.5) -- cycle;	
					
					\foreach \x in {2}
					\foreach \y in {2,3,4,7,8,9,10}
					\draw (\x-0.25,\y-0.25) -- (\x+0.25,\y-0.25) -- (\x+0.25,\y+0.25) -- (\x-0.25,\y+0.25) -- cycle;
					
					\draw [semithick,->,>=stealth] (2,4) -- (2,7+0.1);
				\end{scope}
			\end{tikzpicture}
			\caption {Variable and wire gadget}
			\label {Fig_variable_wire_gadget}
		\end{figure}

		Encoding the truth values so that 1 square corresponds to 1 and 3 squares to 0 may seem somewhat counter-intuitive at first, but it will be explained further in the clause gadget. In this framework, we interpret 1 as an ``open lock'' and 0 as a ``closed lock''.
		
		The arrow itself functions as a wire gadget, transmitting the signal from start to finish. To prevent interference from other blocks, we fence each arrow on both sides by shaded cells. Consequently, each arrow has space for only two blocks: one at the beginning and one at the end of the arrow.
		
		\subsection*{Negation gadget}
		
		The negation gadget, which inverts the value of a Boolean variable, is shown in Fig.~\ref{Fig_negation_gadget} (a).
		It consists of two variable gadgets representing $x$ and $\bar{x}$ with a single square positioned between them.
		Since each block must contain exactly one square placed on an arrow,
		this intermediate square will belong either to the $x$-block (Fig.~\ref{Fig_negation_gadget} (b)), or to the $\bar{x}$-block (Fig.~\ref{Fig_negation_gadget} (c)) in any valid solution, thereby enforcing the inversion of the signal.
		
		\begin{figure}[p]
			\centering
			\begin{tikzpicture}[scale=0.40]
				
				
				\foreach \x in {0,...,3}
				\draw (\x+0.5,0.5)--(\x+0.5,8.5);
				
				\foreach \y in {0,...,8}
				\draw (0.5,\y+0.5)--(3.5,\y+0.5);
				
				\foreach \x in {1,3}	
				\foreach \y in {1,...,8}
				\draw [fill = MidnightBlue!30] (\x-0.5,\y-0.5) -- (\x+0.5,\y-0.5) -- (\x+0.5,\y+0.5) -- (\x-0.5,\y+0.5) -- cycle;

				\foreach \x in {2}
				\foreach \y in {2,3,5,7}
				\draw (\x-0.25,\y-0.25) -- (\x+0.25,\y-0.25) -- (\x+0.25,\y+0.25) -- (\x-0.25,\y+0.25) -- cycle;
				
				\draw [semithick,->,>=stealth] (2,0.5) -- (2,2+0.1);
				\draw [semithick,->,>=stealth] (2,7) -- (2,8+0.75);
				
				\node at (2,0) {$x$};
				\node at (2,9.15) {$\bar{x}$};
				
				\node at (2,-1) {(a)};
				
				\begin{scope}[xshift=8cm]
					
					
					\foreach \x in {0,...,3}
					\draw (\x+0.5,0.5)--(\x+0.5,8.5);
					
					\foreach \y in {0,...,8}
					\draw (0.5,\y+0.5)--(3.5,\y+0.5);
					
					\foreach \x in {1,3}	
					\foreach \y in {1,...,8}
					\draw [fill = MidnightBlue!30] (\x-0.5,\y-0.5) -- 	(\x+0.5,\y-0.5) -- (\x+0.5,\y+0.5) -- (\x-0.5,\y+0.5) -- cycle;

					\foreach \x in {2}
					\foreach \y in {2,3,5,6,7}
					\draw (\x-0.25,\y-0.25) -- (\x+0.25,\y-0.25) -- 	(\x+0.25,\y+0.25) -- (\x-0.25,\y+0.25) -- cycle;
					
					\draw [semithick,->,>=stealth] (2,0.5) -- (2,2+0.1);
					\draw [semithick,->,>=stealth] (2,7) -- (2,8+0.75);
					
					\node at (2,0) {\bf{1}};	
					\node at (2,9.15) {\bf{0}};
					
					\node at (2,-1) {(b)};
					
				\end{scope}
				
				\begin{scope}[xshift=16cm]
					
					
					\foreach \x in {0,...,3}
					\draw (\x+0.5,0.5)--(\x+0.5,8.5);
					
					\foreach \y in {0,...,8}
					\draw (0.5,\y+0.5)--(3.5,\y+0.5);
					
					\foreach \x in {1,3}	
					\foreach \y in {1,...,8}
					\draw [fill = MidnightBlue!30] (\x-0.5,\y-0.5) -- (\x+0.5,\y-0.5) -- (\x+0.5,\y+0.5) -- (\x-0.5,\y+0.5) -- cycle;

					\foreach \x in {2}
					\foreach \y in {2,3,4,5,7}
					\draw (\x-0.25,\y-0.25) -- (\x+0.25,\y-0.25) -- (\x+0.25,\y+0.25) -- (\x-0.25,\y+0.25) -- cycle;
					
					\draw [semithick,->,>=stealth] (2,0.5) -- (2,2+0.1);
					\draw [semithick,->,>=stealth] (2,7) -- (2,8+0.75);
					
					\node at (2,0) {\bf{0}};	
					\node at (2,9.15) {\bf{1}};
					
					\node at (2,-1) {(c)};
					
				\end{scope}
			\end{tikzpicture}
			\caption {Negation gadget}
			\label {Fig_negation_gadget}
		\end{figure}
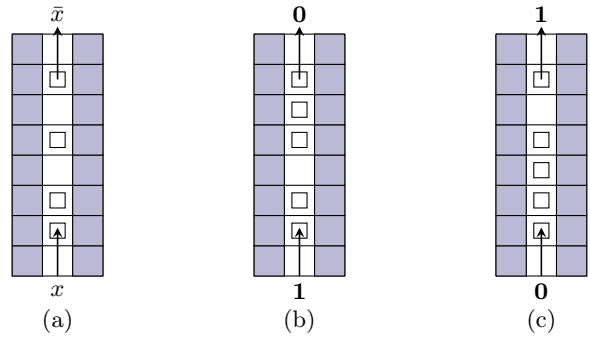

		\subsection*{Split gadget}
		
		The split gadget shown in Fig.~\ref{Fig_split_gadget} (a) is used to duplicate a signal when a single variable appears in multiple clauses.
		As before, since each block contains exactly one square placed on an arrow and the block shape is preserved in the direction of the arrow without rotating and flipping, the puzzle admits only two feasible solutions: Fig.~\ref{Fig_split_gadget} (b) for splitting the $x=1$ signal and Fig.~\ref{Fig_split_gadget} (c) for splitting $x=0$.
		
		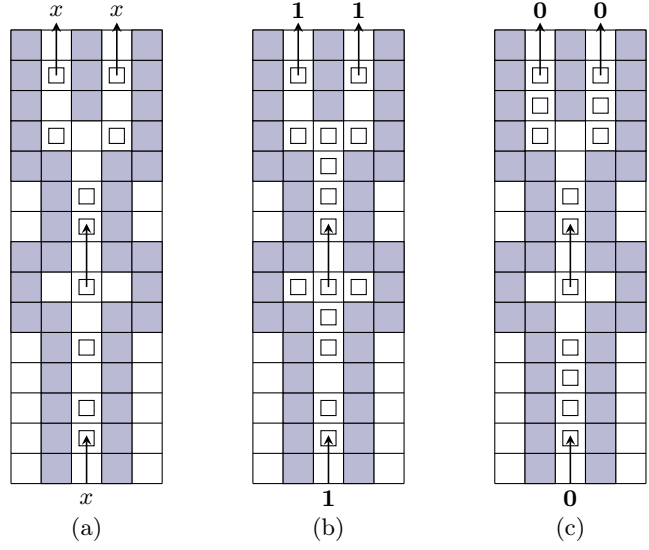
\begin{figure}[p]
			\centering
			\begin{tikzpicture}[scale=0.40]
				
				
				\node at (3,-1) {(a)};
				
				\node at (3,0) {$x$};
				\node at (2,16.15) {$x$};
				\node at (4,16.15) {$x$};
				
				\foreach \x in {0,...,5}
				\draw (\x+0.5,0.5)--(\x+0.5,15.5);
				
				\foreach \y in {0,...,15}
				\draw (0.5,\y+0.5)--(5.5,\y+0.5);
				
				\foreach \x in {2,4}	
				\foreach \y in {1,...,6,8,9,10,11}
				\draw [fill = MidnightBlue!30] (\x-0.5,\y-0.5) -- (\x+0.5,\y-0.5) -- (\x+0.5,\y+0.5) -- (\x-0.5,\y+0.5) -- cycle;	
				
				\foreach \x in {1,5}	
				\foreach \y in {6,7,8,11,12,...,15}
				\draw [fill = MidnightBlue!30] (\x-0.5,\y-0.5) -- (\x+0.5,\y-0.5) -- (\x+0.5,\y+0.5) -- (\x-0.5,\y+0.5) -- cycle;	
				
				\foreach \x in {3}	
				\foreach \y in {13,...,15}
				\draw [fill = MidnightBlue!30] (\x-0.5,\y-0.5) -- (\x+0.5,\y-0.5) -- (\x+0.5,\y+0.5) -- (\x-0.5,\y+0.5) -- cycle;	
				
				\foreach \x in {3}
				\foreach \y in {2,3,5,7,9,10}
				\draw (\x-0.25,\y-0.25) -- (\x+0.25,\y-0.25) -- (\x+0.25,\y+0.25) -- (\x-0.25,\y+0.25) -- cycle;
				
				\foreach \x in {2,4}
				\foreach \y in {12,14}
				\draw (\x-0.25,\y-0.25) -- (\x+0.25,\y-0.25) -- (\x+0.25,\y+0.25) -- (\x-0.25,\y+0.25) -- cycle;
				
				\draw [semithick,->,>=stealth] (3,0.5) -- (3,2+0.1);
				\draw [semithick,->,>=stealth] (3,7) -- (3,9+0.1);
				\draw [semithick,->,>=stealth] (2,14) -- (2,15+0.75);
				\draw [semithick,->,>=stealth] (4,14) -- (4,15+0.75);

				\begin{scope}[xshift=8cm]
					
					\node at (3,-1) {(b)};
					
					\node at (3,0) {\bf{1}};
					\node at (2,16.15) {\bf{1}};
					\node at (4,16.15) {\bf{1}};
					
					\foreach \x in {0,...,5}
					\draw (\x+0.5,0.5)--(\x+0.5,15.5);
					
					\foreach \y in {0,...,15}
					\draw (0.5,\y+0.5)--(5.5,\y+0.5);
					
					\foreach \x in {2,4}	
					\foreach \y in {1,...,6,8,9,10,11}
					\draw [fill = MidnightBlue!30] (\x-0.5,\y-0.5) -- (\x+0.5,\y-0.5) -- (\x+0.5,\y+0.5) -- (\x-0.5,\y+0.5) -- cycle;	
					
					\foreach \x in {1,5}	
					\foreach \y in {6,7,8,11,12,...,15}
					\draw [fill = MidnightBlue!30] (\x-0.5,\y-0.5) -- (\x+0.5,\y-0.5) -- (\x+0.5,\y+0.5) -- (\x-0.5,\y+0.5) -- cycle;	
					
					\foreach \x in {3}	
					\foreach \y in {13,...,15}
					\draw [fill = MidnightBlue!30] (\x-0.5,\y-0.5) -- (\x+0.5,\y-0.5) -- (\x+0.5,\y+0.5) -- (\x-0.5,\y+0.5) -- cycle;	
					
					\foreach \x in {3}
					\foreach \y in {2,3,5,6,7,9,10,11,12}
					\draw (\x-0.25,\y-0.25) -- (\x+0.25,\y-0.25) -- (\x+0.25,\y+0.25) -- (\x-0.25,\y+0.25) -- cycle;

					\foreach \x in {2,4}
					\foreach \y in {7,12,14}
					\draw (\x-0.25,\y-0.25) -- (\x+0.25,\y-0.25) -- (\x+0.25,\y+0.25) -- (\x-0.25,\y+0.25) -- cycle;
					
					\draw [semithick,->,>=stealth] (3,0.5) -- (3,2+0.1);
					\draw [semithick,->,>=stealth] (3,7) -- (3,9+0.1);
					\draw [semithick,->,>=stealth] (2,14) -- (2,15+0.75);
					\draw [semithick,->,>=stealth] (4,14) -- (4,15+0.75);
				\end{scope}
				
				\begin{scope}[xshift=16cm]
					
					\node at (3,-1) {(c)};
					
					\node at (3,0) {\bf {0}};
					\node at (2,16.15) {\bf{0}};
					\node at (4,16.15) {\bf{0}};
					
					\foreach \x in {0,...,5}
					\draw (\x+0.5,0.5)--(\x+0.5,15.5);
					
					\foreach \y in {0,...,15}
					\draw (0.5,\y+0.5)--(5.5,\y+0.5);
					
					\foreach \x in {2,4}	
					\foreach \y in {1,...,6,8,9,10,11}
					\draw [fill = MidnightBlue!30] (\x-0.5,\y-0.5) -- 	(\x+0.5,\y-0.5) -- (\x+0.5,\y+0.5) -- (\x-0.5,\y+0.5) -- cycle;	
					
					\foreach \x in {1,5}	
					\foreach \y in {6,7,8,11,12,...,15}
					\draw [fill = MidnightBlue!30] (\x-0.5,\y-0.5) -- 	(\x+0.5,\y-0.5) -- (\x+0.5,\y+0.5) -- (\x-0.5,\y+0.5) -- cycle;	
					
					\foreach \x in {3}	
					\foreach \y in {13,...,15}
					\draw [fill = MidnightBlue!30] (\x-0.5,\y-0.5) -- 	(\x+0.5,\y-0.5) -- (\x+0.5,\y+0.5) -- (\x-0.5,\y+0.5) -- cycle;	
					
					\foreach \x in {3}
					\foreach \y in {2,3,4,5,7,9,10}
					\draw (\x-0.25,\y-0.25) -- (\x+0.25,\y-0.25) -- 	(\x+0.25,\y+0.25) -- (\x-0.25,\y+0.25) -- cycle;
					
					\foreach \x in {2,4}
					\foreach \y in {12,13,14}
					\draw (\x-0.25,\y-0.25) -- (\x+0.25,\y-0.25) -- 	(\x+0.25,\y+0.25) -- (\x-0.25,\y+0.25) -- cycle;
					
					\draw [semithick,->,>=stealth] (3,0.5) -- (3,2+0.1);
					\draw [semithick,->,>=stealth] (3,7) -- (3,9+0.1);
					\draw [semithick,->,>=stealth] (2,14) -- (2,15+0.75);
					\draw [semithick,->,>=stealth] (4,14) -- (4,15+0.75);
				\end{scope}
			\end{tikzpicture}
			\caption {Split gadget}
			\label {Fig_split_gadget}
		\end{figure}
		
		\subsection*{Clause gadget}
		
		The structure of the clause gadget for a clause $C=\{x \vee y\vee z\}$ is shown in Fig.~\ref{Fig_clause_gadget}. Since there is a vertical line block of 5 squares at the end of the horizontal arrow, the puzzle rules require that it be preceded by a vertical line block of 4 squares.
		\begin{figure}[p]
			\centering
			\begin{tikzpicture}[scale=0.40]
				
				
				
				\foreach \x in {0,...,9}
				\draw (\x+0.5,0.5)--(\x+0.5,9.5);
				
				\foreach \y in {0,...,9}
				\draw (0.5,\y+0.5)--(9.5,\y+0.5);
				
				\foreach \x in {1}	
				\foreach \y in {1,...,9}
				\draw [fill = MidnightBlue!30] (\x-0.5,\y-0.5) -- (\x+0.5,\y-0.5) -- (\x+0.5,\y+0.5) -- (\x-0.5,\y+0.5) -- cycle;	
				
				\foreach \x in {3,5,7}	
				\foreach \y in {1,...,5,7,8,9}
				\draw [fill = MidnightBlue!30] (\x-0.5,\y-0.5) -- (\x+0.5,\y-0.5) -- (\x+0.5,\y+0.5) -- (\x-0.5,\y+0.5) -- cycle;	
				
				\foreach \x in {2,4,6}	
				\foreach \y in {9}
				\draw [fill = MidnightBlue!30] (\x-0.5,\y-0.5) -- (\x+0.5,\y-0.5) -- (\x+0.5,\y+0.5) -- (\x-0.5,\y+0.5) -- cycle;	
				
				\foreach \x in {8}	
				\foreach \y in {3,9}
				\draw [fill = MidnightBlue!30] (\x-0.5,\y-0.5) -- (\x+0.5,\y-0.5) -- (\x+0.5,\y+0.5) -- (\x-0.5,\y+0.5) -- cycle;	
				
				\foreach \x in {9}	
				\foreach \y in {3,...,9}
				\draw [fill = MidnightBlue!30] (\x-0.5,\y-0.5) -- (\x+0.5,\y-0.5) -- (\x+0.5,\y+0.5) -- (\x-0.5,\y+0.5) -- cycle;

				\foreach \x in {2,4,6}
				\foreach \y in {2,3,5}
				\draw (\x-0.25,\y-0.25) -- (\x+0.25,\y-0.25) -- (\x+0.25,\y+0.25) -- (\x-0.25,\y+0.25) -- cycle;
				
				\foreach \x in {8}
				\foreach \y in {4,...,8}
				\draw (\x-0.25,\y-0.25) -- (\x+0.25,\y-0.25) -- (\x+0.25,\y+0.25) -- (\x-0.25,\y+0.25) -- cycle;
				
				\draw [semithick,->,>=stealth] (2,6) -- (8+0.1,6);
				
				\foreach \x in {2,4,6}
				\draw [semithick,->,>=stealth] (\x,0.5) -- (\x,2+0.1);
				
				\node at (2,0) {$x$};
				\node at (4,0) {$y$};
				\node at (6,0) {$z$};
				
				\node at (2,10) {$x$};
				\node at (3,10) {$\vee$};
				\node at (4,10) {$y$};
				\node at (5,10) {$\vee$};
				\node at (6,10) {$z$};
				
			\end{tikzpicture}
			\caption {Clause gadget}
			\label {Fig_clause_gadget}
		\end{figure}
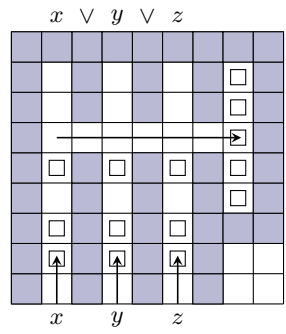
		
		It is easy to see that such a block can only be placed in one of the three columns corresponding to the literals $x$, $y$, and $z$.
		However, if a false signal is received for a literal, then the corresponding column is locked from below, because the same block cannot be placed on two different arrows simultaneously.
		Therefore, the clause gadget has a solution if and only if at least one of the literals $x$, $y$, or $z$ is satisfied.
		Examples of feasible cases with one (a), two (b), and three (c) true literals, as well as an infeasible case when all literals are false (d), are shown in Fig.~\ref{Fig_clause_gadget_feasible}.

		\subsection*{Crossover gadget}
		
		Signals coming from variable gadgets to clause gadgets may cross paths along their routes. To address this issue and ensure correct signal transmission without interference or leakage at wire intersections, we introduce the crossover gadget shown in Fig.~\ref{Fig_crossover_gadget}.

		\begin{figure}[p]
			\centering
			\begin{tikzpicture}[scale=0.40]
				
				
				
				\foreach \x in {0,...,9}
				\draw (\x+0.5,0.5)--(\x+0.5,9.5);
				
				\foreach \y in {0,...,9}
				\draw (0.5,\y+0.5)--(9.5,\y+0.5);
				
				\foreach \x in {1}	
				\foreach \y in {1,...,9}
				\draw [fill = MidnightBlue!30] (\x-0.5,\y-0.5) -- (\x+0.5,\y-0.5) -- (\x+0.5,\y+0.5) -- (\x-0.5,\y+0.5) -- cycle;	
				
				\foreach \x in {3,5,7}	
				\foreach \y in {1,...,5,7,8,9}
				\draw [fill = MidnightBlue!30] (\x-0.5,\y-0.5) -- (\x+0.5,\y-0.5) -- (\x+0.5,\y+0.5) -- (\x-0.5,\y+0.5) -- cycle;	
				
				\foreach \x in {2,4,6}	
				\foreach \y in {9}
				\draw [fill = MidnightBlue!30] (\x-0.5,\y-0.5) -- (\x+0.5,\y-0.5) -- (\x+0.5,\y+0.5) -- (\x-0.5,\y+0.5) -- cycle;	
				
				\foreach \x in {8}	
				\foreach \y in {3,9}
				\draw [fill = MidnightBlue!30] (\x-0.5,\y-0.5) -- (\x+0.5,\y-0.5) -- (\x+0.5,\y+0.5) -- (\x-0.5,\y+0.5) -- cycle;	
				
				\foreach \x in {9}	
				\foreach \y in {3,...,9}
				\draw [fill = MidnightBlue!30] (\x-0.5,\y-0.5) -- (\x+0.5,\y-0.5) -- (\x+0.5,\y+0.5) -- (\x-0.5,\y+0.5) -- cycle;

				\foreach \x in {2}
				\foreach \y in {2,3,5,6,7,8}
				\draw (\x-0.25,\y-0.25) -- (\x+0.25,\y-0.25) -- (\x+0.25,\y+0.25) -- (\x-0.25,\y+0.25) -- cycle;
				
				\foreach \x in {4,6}
				\foreach \y in {2,3,4,5}
				\draw (\x-0.25,\y-0.25) -- (\x+0.25,\y-0.25) -- (\x+0.25,\y+0.25) -- (\x-0.25,\y+0.25) -- cycle;
				
				\foreach \x in {8}
				\foreach \y in {4,...,8}
				\draw (\x-0.25,\y-0.25) -- (\x+0.25,\y-0.25) -- (\x+0.25,\y+0.25) -- (\x-0.25,\y+0.25) -- cycle;
				
				\draw [semithick,->,>=stealth] (2,6) -- (8+0.1,6);
				
				\foreach \x in {2,4,6}
				\draw [semithick,->,>=stealth] (\x,0.5) -- (\x,2+0.1);
				
				\node at (2,0) {\bf{1}};
				\node at (4,0) {\bf{0}};
				\node at (6,0) {\bf{0}};
				
				\node at (5,-1) {(a)};

				\begin{scope}[xshift=14cm]
					
					
					\foreach \x in {0,...,9}
					\draw (\x+0.5,0.5)--(\x+0.5,9.5);
					
					\foreach \y in {0,...,9}
					\draw (0.5,\y+0.5)--(9.5,\y+0.5);
					
					\foreach \x in {1}	
					\foreach \y in {1,...,9}
					\draw [fill = MidnightBlue!30] (\x-0.5,\y-0.5) -- (\x+0.5,\y-0.5) -- (\x+0.5,\y+0.5) -- (\x-0.5,\y+0.5) -- cycle;	
					
					\foreach \x in {3,5,7}	
					\foreach \y in {1,...,5,7,8,9}
					\draw [fill = MidnightBlue!30] (\x-0.5,\y-0.5) -- (\x+0.5,\y-0.5) -- (\x+0.5,\y+0.5) -- (\x-0.5,\y+0.5) -- cycle;	
					
					\foreach \x in {2,4,6}	
					\foreach \y in {9}
					\draw [fill = MidnightBlue!30] (\x-0.5,\y-0.5) -- (\x+0.5,\y-0.5) -- (\x+0.5,\y+0.5) -- (\x-0.5,\y+0.5) -- cycle;	
					
					\foreach \x in {8}	
					\foreach \y in {3,9}
					\draw [fill = MidnightBlue!30] (\x-0.5,\y-0.5) -- (\x+0.5,\y-0.5) -- (\x+0.5,\y+0.5) -- (\x-0.5,\y+0.5) -- cycle;	
					
					\foreach \x in {9}	
					\foreach \y in {3,...,9}
					\draw [fill = MidnightBlue!30] (\x-0.5,\y-0.5) -- (\x+0.5,\y-0.5) -- (\x+0.5,\y+0.5) -- (\x-0.5,\y+0.5) -- cycle;	
					
					\foreach \x in {2}
					\foreach \y in {2,3,5,6,7}
					\draw (\x-0.25,\y-0.25) -- (\x+0.25,\y-0.25) -- (\x+0.25,\y+0.25) -- (\x-0.25,\y+0.25) -- cycle;
					
					\foreach \x in {4}
					\foreach \y in {2,3,5,6,7,8}
					\draw (\x-0.25,\y-0.25) -- (\x+0.25,\y-0.25) -- (\x+0.25,\y+0.25) -- (\x-0.25,\y+0.25) -- cycle;
					
					\foreach \x in {6}
					\foreach \y in {2,3,4,5}
					\draw (\x-0.25,\y-0.25) -- (\x+0.25,\y-0.25) -- (\x+0.25,\y+0.25) -- (\x-0.25,\y+0.25) -- cycle;
					
					\foreach \x in {8}
					\foreach \y in {4,...,8}
					\draw (\x-0.25,\y-0.25) -- (\x+0.25,\y-0.25) -- (\x+0.25,\y+0.25) -- (\x-0.25,\y+0.25) -- cycle;
					
					\draw [semithick,->,>=stealth] (2,6) -- (8+0.1,6);
					
					\foreach \x in {2,4,6}
					\draw [semithick,->,>=stealth] (\x,0.5) -- (\x,2+0.1);
					
					\node at (2,0) {\bf{1}};
					\node at (4,0) {\bf{1}};
					\node at (6,0) {\bf{0}};
					
					\node at (5,-1) {(b)};
				\end{scope}

				\begin{scope}[yshift=-13cm]
					
					
					\foreach \x in {0,...,9}
					\draw (\x+0.5,0.5)--(\x+0.5,9.5);
					
					\foreach \y in {0,...,9}
					\draw (0.5,\y+0.5)--(9.5,\y+0.5);
					
					\foreach \x in {1}	
					\foreach \y in {1,...,9}
					\draw [fill = MidnightBlue!30] (\x-0.5,\y-0.5) -- (\x+0.5,\y-0.5) -- (\x+0.5,\y+0.5) -- (\x-0.5,\y+0.5) -- cycle;	
					
					\foreach \x in {3,5,7}	
					\foreach \y in {1,...,5,7,8,9}
					\draw [fill = MidnightBlue!30] (\x-0.5,\y-0.5) -- (\x+0.5,\y-0.5) -- (\x+0.5,\y+0.5) -- (\x-0.5,\y+0.5) -- cycle;	
					
					\foreach \x in {2,4,6}	
					\foreach \y in {9}
					\draw [fill = MidnightBlue!30] (\x-0.5,\y-0.5) -- (\x+0.5,\y-0.5) -- (\x+0.5,\y+0.5) -- (\x-0.5,\y+0.5) -- cycle;	
					
					\foreach \x in {8}	
					\foreach \y in {3,9}
					\draw [fill = MidnightBlue!30] (\x-0.5,\y-0.5) -- (\x+0.5,\y-0.5) -- (\x+0.5,\y+0.5) -- (\x-0.5,\y+0.5) -- cycle;	
					
					\foreach \x in {9}	
					\foreach \y in {3,...,9}
					\draw [fill = MidnightBlue!30] (\x-0.5,\y-0.5) -- (\x+0.5,\y-0.5) -- (\x+0.5,\y+0.5) -- (\x-0.5,\y+0.5) -- cycle;

					\foreach \x in {2}
					\foreach \y in {2,3,5,6}
					\draw (\x-0.25,\y-0.25) -- (\x+0.25,\y-0.25) -- (\x+0.25,\y+0.25) -- (\x-0.25,\y+0.25) -- cycle;
					
					\foreach \x in {4}
					\foreach \y in {2,3,5,6,7}
					\draw (\x-0.25,\y-0.25) -- (\x+0.25,\y-0.25) -- (\x+0.25,\y+0.25) -- (\x-0.25,\y+0.25) -- cycle;
					
					\foreach \x in {6}
					\foreach \y in {2,3,5,6,7,8}
					\draw (\x-0.25,\y-0.25) -- (\x+0.25,\y-0.25) -- (\x+0.25,\y+0.25) -- (\x-0.25,\y+0.25) -- cycle;
					
					\foreach \x in {8}
					\foreach \y in {4,...,8}
					\draw (\x-0.25,\y-0.25) -- (\x+0.25,\y-0.25) -- (\x+0.25,\y+0.25) -- (\x-0.25,\y+0.25) -- cycle;
					
					\draw [semithick,->,>=stealth] (2,6) -- (8+0.1,6);
					
					\foreach \x in {2,4,6}
					\draw [semithick,->,>=stealth] (\x,0.5) -- (\x,2+0.1);
					
					\node at (2,0) {\bf{1}};
					\node at (4,0) {\bf{1}};
					\node at (6,0) {\bf{1}};
					
					\node at (5,-1) {(c)};
				\end{scope}

				\begin{scope}[xshift=14cm,yshift=-13cm]
					
					
					\foreach \x in {0,...,9}
					\draw (\x+0.5,0.5)--(\x+0.5,9.5);
					
					\foreach \y in {0,...,9}
					\draw (0.5,\y+0.5)--(9.5,\y+0.5);
					
					\foreach \x in {1}	
					\foreach \y in {1,...,9}
					\draw [fill = MidnightBlue!30] (\x-0.5,\y-0.5) -- (\x+0.5,\y-0.5) -- (\x+0.5,\y+0.5) -- (\x-0.5,\y+0.5) -- cycle;	
					
					\foreach \x in {3,5,7}	
					\foreach \y in {1,...,5,7,8,9}
					\draw [fill = MidnightBlue!30] (\x-0.5,\y-0.5) -- (\x+0.5,\y-0.5) -- (\x+0.5,\y+0.5) -- (\x-0.5,\y+0.5) -- cycle;	
					
					\foreach \x in {2,4,6}	
					\foreach \y in {9}
					\draw [fill = MidnightBlue!30] (\x-0.5,\y-0.5) -- (\x+0.5,\y-0.5) -- (\x+0.5,\y+0.5) -- (\x-0.5,\y+0.5) -- cycle;	
					
					\foreach \x in {8}	
					\foreach \y in {3,9}
					\draw [fill = MidnightBlue!30] (\x-0.5,\y-0.5) -- (\x+0.5,\y-0.5) -- (\x+0.5,\y+0.5) -- (\x-0.5,\y+0.5) -- cycle;	
					
					\foreach \x in {9}	
					\foreach \y in {3,...,9}
					\draw [fill = MidnightBlue!30] (\x-0.5,\y-0.5) -- (\x+0.5,\y-0.5) -- (\x+0.5,\y+0.5) -- (\x-0.5,\y+0.5) -- cycle;

					\foreach \x in {2,4,6}
					\foreach \y in {2,3,4,5}
					\draw (\x-0.25,\y-0.25) -- (\x+0.25,\y-0.25) -- (\x+0.25,\y+0.25) -- (\x-0.25,\y+0.25) -- cycle;
					
					\foreach \x in {8}
					\foreach \y in {4,...,8}
					\draw (\x-0.25,\y-0.25) -- (\x+0.25,\y-0.25) -- (\x+0.25,\y+0.25) -- (\x-0.25,\y+0.25) -- cycle;
					
					\draw [semithick,->,>=stealth] (2,6) -- (8+0.1,6);
					
					\foreach \x in {2,4,6}
					\draw [semithick,->,>=stealth] (\x,0.5) -- (\x,2+0.1);
					
					\node at (2,0) {\bf{0}};
					\node at (4,0) {\bf{0}};
					\node at (6,0) {\bf{0}};
					
					\node at (5,-1) {(d)};
				\end{scope}
				
			\end{tikzpicture}
			\caption {Three feasible cases (a), (b), (c), and one infeasible case (d) of a clause gadget}
			\label {Fig_clause_gadget_feasible}
		\end{figure}
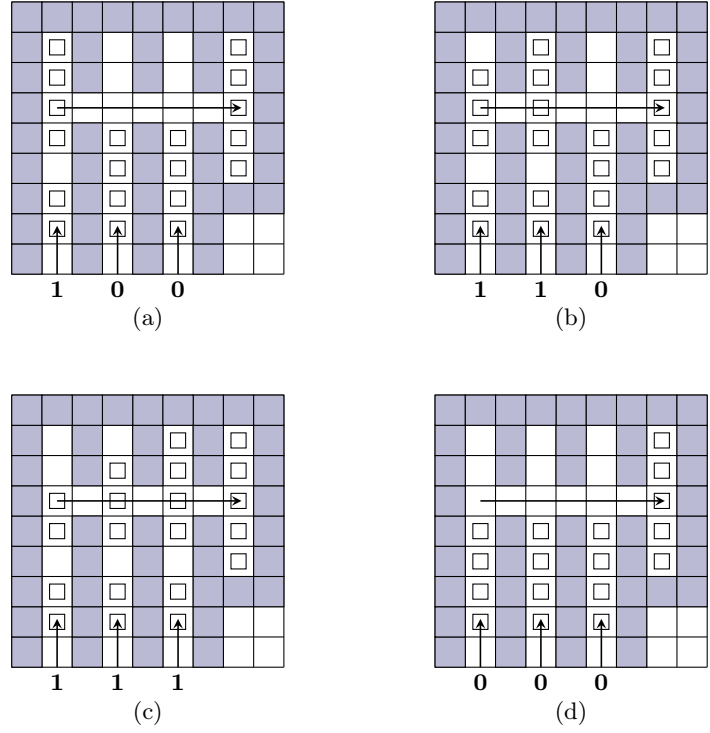

		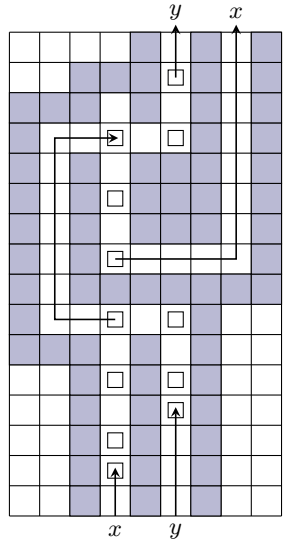
\begin{figure}[p]
			\centering
			\begin{tikzpicture}[scale=0.40]
				
				
				
				\foreach \x in {0,...,9}
				\draw (\x+0.5,0.5)--(\x+0.5,16.5);
				
				\foreach \y in {0,...,16}
				\draw (0.5,\y+0.5)--(9.5,\y+0.5);

				\foreach \x in {4}
				\foreach \y in {2,3,5,7,9,11,13}
				\draw (\x-0.25,\y-0.25) -- (\x+0.25,\y-0.25) -- (\x+0.25,\y+0.25) -- (\x-0.25,\y+0.25) -- cycle;
				
				\foreach \x in {6}
				\foreach \y in {4,5,7,13,15}
				\draw (\x-0.25,\y-0.25) -- (\x+0.25,\y-0.25) -- (\x+0.25,\y+0.25) -- (\x-0.25,\y+0.25) -- cycle;
				
				\draw [semithick,->,>=stealth] (4,0.5) -- (4,2+0.1);
				\draw [semithick,->,>=stealth] (6,0.5) -- (6,4+0.1);
				
				\draw [semithick,->,>=stealth] (4,7) -- (2,7) -- (2,13) -- (4+0.1,13);
				
				\draw [semithick,->,>=stealth] (4,9) -- (8,9) -- (8,16+0.75);
				\draw [semithick,->,>=stealth] (6,15) -- (6,16+0.75);
				
				\node at (6,17.15) {$y$};
				\node at (8,17.15) {$x$};
				
				\node at (4,0) {$x$};
				\node at (6,0) {$y$};

				\foreach \x in {1}	
				\foreach \y in {6,...,14}
				\draw [fill = MidnightBlue!30] (\x-0.5,\y-0.5) -- (\x+0.5,\y-0.5) -- (\x+0.5,\y+0.5) -- (\x-0.5,\y+0.5) -- cycle;	
				
				\foreach \x in {2}	
				\foreach \y in {6,14}
				\draw [fill = MidnightBlue!30] (\x-0.5,\y-0.5) -- (\x+0.5,\y-0.5) -- (\x+0.5,\y+0.5) -- (\x-0.5,\y+0.5) -- cycle;	
				
				\foreach \x in {3}	
				\foreach \y in {1,...,6,8,9,10,11,12,14,15}
				\draw [fill = MidnightBlue!30] (\x-0.5,\y-0.5) -- (\x+0.5,\y-0.5) -- (\x+0.5,\y+0.5) -- (\x-0.5,\y+0.5) -- cycle;	
				
				\foreach \x in {4}	
				\foreach \y in {8,15}
				\draw [fill = MidnightBlue!30] (\x-0.5,\y-0.5) -- (\x+0.5,\y-0.5) -- (\x+0.5,\y+0.5) -- (\x-0.5,\y+0.5) -- cycle;	
				
				\foreach \x in {5}	
				\foreach \y in {1,...,6,8,10,11,12,14,15,16}
				\draw [fill = MidnightBlue!30] (\x-0.5,\y-0.5) -- (\x+0.5,\y-0.5) -- (\x+0.5,\y+0.5) -- (\x-0.5,\y+0.5) -- cycle;	
				
				\foreach \x in {6}	
				\foreach \y in {8,10,11,12}
				\draw [fill = MidnightBlue!30] (\x-0.5,\y-0.5) -- (\x+0.5,\y-0.5) -- (\x+0.5,\y+0.5) -- (\x-0.5,\y+0.5) -- cycle;	
				
				\foreach \x in {7}	
				\foreach \y in {1,...,8,10,11,12,13,14,15,16}
				\draw [fill = MidnightBlue!30] (\x-0.5,\y-0.5) -- (\x+0.5,\y-0.5) -- (\x+0.5,\y+0.5) -- (\x-0.5,\y+0.5) -- cycle;	
				
				\foreach \x in {8}	
				\foreach \y in {8}
				\draw [fill = MidnightBlue!30] (\x-0.5,\y-0.5) -- (\x+0.5,\y-0.5) -- (\x+0.5,\y+0.5) -- (\x-0.5,\y+0.5) -- cycle;	
				
				\foreach \x in {9}	
				\foreach \y in {8,...,16}
				\draw [fill = MidnightBlue!30] (\x-0.5,\y-0.5) -- (\x+0.5,\y-0.5) -- (\x+0.5,\y+0.5) -- (\x-0.5,\y+0.5) -- cycle;	
				
			\end{tikzpicture}
			\caption {Crossover gadget}
			\label {Fig_crossover_gadget}
		\end{figure}
		
		Here, to cross the signals, we use a $\mathrm{\Gamma}$-shaped polyomino, where the vertical block corresponds to signal $x$, and the horizontal one to $y$. The trick is that according to the rules of the \textsc{Evolomino} puzzle, consecutive blocks cannot be rotated in the direction of the arrow: horizontal blocks remain horizontal, and vertical blocks remain vertical. This allows us to preserve the values of the $x$ and $y$ signals, while swapping the order of the wires.
		
		It remains to verify that there are exactly four feasible solutions to the crossover gadget that satisfy the rules of the puzzle, corresponding to the four possible combinations of values for the intersecting signals (see Fig.~\ref{Fig_crossover_gadget_cases}).

		\begin{figure}[p]
			\centering
			\begin{tikzpicture}[scale=0.40]
				
				
				\foreach \x in {0,...,9}
				\draw (\x+0.5,0.5)--(\x+0.5,16.5);
				
				\foreach \y in {0,...,16}
				\draw (0.5,\y+0.5)--(9.5,\y+0.5);
				
				\node at (4,0) {\bf{0}};
				\node at (6,0) {\bf{0}};
				
				\node at (6,17.15) {\bf{0}};
				\node at (8,17.15) {\bf{0}};

				\foreach \x in {4}
				\foreach \y in {2,3,4,5,7,9,10,11,13,14}
				\draw (\x-0.25,\y-0.25) -- (\x+0.25,\y-0.25) -- (\x+0.25,\y+0.25) -- (\x-0.25,\y+0.25) -- cycle;
				
				\foreach \x in {6}
				\foreach \y in {4,5,6,7,13,14,15}
				\draw (\x-0.25,\y-0.25) -- (\x+0.25,\y-0.25) -- (\x+0.25,\y+0.25) -- (\x-0.25,\y+0.25) -- cycle;
				
				\draw [semithick,->,>=stealth] (4,0.5) -- (4,2+0.1);
				\draw [semithick,->,>=stealth] (6,0.5) -- (6,4+0.1);
				
				\draw [semithick,->,>=stealth] (4,7) -- (2,7) -- (2,13) -- (4+0.1,13);
				
				\draw [semithick,->,>=stealth] (4,9) -- (8,9) -- (8,16+0.75);
				\draw [semithick,->,>=stealth] (6,15) -- (6,16+0.75);

				\foreach \x in {1}	
				\foreach \y in {6,...,14}
				\draw [fill = MidnightBlue!30] (\x-0.5,\y-0.5) -- (\x+0.5,\y-0.5) -- (\x+0.5,\y+0.5) -- (\x-0.5,\y+0.5) -- cycle;	
				
				\foreach \x in {2}	
				\foreach \y in {6,14}
				\draw [fill = MidnightBlue!30] (\x-0.5,\y-0.5) -- (\x+0.5,\y-0.5) -- (\x+0.5,\y+0.5) -- (\x-0.5,\y+0.5) -- cycle;	
				
				\foreach \x in {3}	
				\foreach \y in {1,...,6,8,9,10,11,12,14,15}
				\draw [fill = MidnightBlue!30] (\x-0.5,\y-0.5) -- (\x+0.5,\y-0.5) -- (\x+0.5,\y+0.5) -- (\x-0.5,\y+0.5) -- cycle;	
				
				\foreach \x in {4}	
				\foreach \y in {8,15}
				\draw [fill = MidnightBlue!30] (\x-0.5,\y-0.5) -- (\x+0.5,\y-0.5) -- (\x+0.5,\y+0.5) -- (\x-0.5,\y+0.5) -- cycle;	
				
				\foreach \x in {5}	
				\foreach \y in {1,...,6,8,10,11,12,14,15,16}
				\draw [fill = MidnightBlue!30] (\x-0.5,\y-0.5) -- (\x+0.5,\y-0.5) -- (\x+0.5,\y+0.5) -- (\x-0.5,\y+0.5) -- cycle;	
				
				\foreach \x in {6}	
				\foreach \y in {8,10,11,12}
				\draw [fill = MidnightBlue!30] (\x-0.5,\y-0.5) -- (\x+0.5,\y-0.5) -- (\x+0.5,\y+0.5) -- (\x-0.5,\y+0.5) -- cycle;	
				
				\foreach \x in {7}	
				\foreach \y in {1,...,8,10,11,12,13,14,15,16}
				\draw [fill = MidnightBlue!30] (\x-0.5,\y-0.5) -- (\x+0.5,\y-0.5) -- (\x+0.5,\y+0.5) -- (\x-0.5,\y+0.5) -- cycle;	
				
				\foreach \x in {8}	
				\foreach \y in {8}
				\draw [fill = MidnightBlue!30] (\x-0.5,\y-0.5) -- (\x+0.5,\y-0.5) -- (\x+0.5,\y+0.5) -- (\x-0.5,\y+0.5) -- cycle;	
				
				\foreach \x in {9}	
				\foreach \y in {8,...,16}
				\draw [fill = MidnightBlue!30] (\x-0.5,\y-0.5) -- (\x+0.5,\y-0.5) -- (\x+0.5,\y+0.5) -- (\x-0.5,\y+0.5) -- cycle;

				\begin{scope}[xshift=14cm]
					
					\foreach \x in {0,...,9}
					\draw (\x+0.5,0.5)--(\x+0.5,16.5);
					
					\foreach \y in {0,...,16}
					\draw (0.5,\y+0.5)--(9.5,\y+0.5);
					
					\node at (4,0) {\bf{0}};
					\node at (6,0) {\bf{1}};
					
					\node at (6,17.15) {\bf{1}};
					\node at (8,17.15) {\bf{0}};

					\foreach \x in {4}
					\foreach \y in {2,3,4,5,7,9,10,11,13,14}
					\draw (\x-0.25,\y-0.25) -- (\x+0.25,\y-0.25) -- (\x+0.25,\y+0.25) -- (\x-0.25,\y+0.25) -- cycle;
					
					\foreach \x in {5}
					\foreach \y in {7,13}
					\draw (\x-0.25,\y-0.25) -- (\x+0.25,\y-0.25) -- (\x+0.25,\y+0.25) -- (\x-0.25,\y+0.25) -- cycle;
					
					\foreach \x in {6}
					\foreach \y in {4,5,7,13,15}
					\draw (\x-0.25,\y-0.25) -- (\x+0.25,\y-0.25) -- (\x+0.25,\y+0.25) -- (\x-0.25,\y+0.25) -- cycle;
					
					\draw [semithick,->,>=stealth] (4,0.5) -- (4,2+0.1);
					\draw [semithick,->,>=stealth] (6,0.5) -- (6,4+0.1);
					
					\draw [semithick,->,>=stealth] (4,7) -- (2,7) -- (2,13) -- (4+0.1,13);
					
					\draw [semithick,->,>=stealth] (4,9) -- (8,9) -- (8,16+0.75);
					\draw [semithick,->,>=stealth] (6,15) -- (6,16+0.75);

					\foreach \x in {1}	
					\foreach \y in {6,...,14}
					\draw [fill = MidnightBlue!30] (\x-0.5,\y-0.5) -- (\x+0.5,\y-0.5) -- (\x+0.5,\y+0.5) -- (\x-0.5,\y+0.5) -- cycle;	
					
					\foreach \x in {2}	
					\foreach \y in {6,14}
					\draw [fill = MidnightBlue!30] (\x-0.5,\y-0.5) -- (\x+0.5,\y-0.5) -- (\x+0.5,\y+0.5) -- (\x-0.5,\y+0.5) -- cycle;	
					
					\foreach \x in {3}	
					\foreach \y in {1,...,6,8,9,10,11,12,14,15}
					\draw [fill = MidnightBlue!30] (\x-0.5,\y-0.5) -- (\x+0.5,\y-0.5) -- (\x+0.5,\y+0.5) -- (\x-0.5,\y+0.5) -- cycle;	
					
					\foreach \x in {4}	
					\foreach \y in {8,15}
					\draw [fill = MidnightBlue!30] (\x-0.5,\y-0.5) -- (\x+0.5,\y-0.5) -- (\x+0.5,\y+0.5) -- (\x-0.5,\y+0.5) -- cycle;	
					
					\foreach \x in {5}	
					\foreach \y in {1,...,6,8,10,11,12,14,15,16}
					\draw [fill = MidnightBlue!30] (\x-0.5,\y-0.5) -- (\x+0.5,\y-0.5) -- (\x+0.5,\y+0.5) -- (\x-0.5,\y+0.5) -- cycle;	
					
					\foreach \x in {6}	
					\foreach \y in {8,10,11,12}
					\draw [fill = MidnightBlue!30] (\x-0.5,\y-0.5) -- (\x+0.5,\y-0.5) -- (\x+0.5,\y+0.5) -- (\x-0.5,\y+0.5) -- cycle;	
					
					\foreach \x in {7}	
					\foreach \y in {1,...,8,10,11,12,13,14,15,16}
					\draw [fill = MidnightBlue!30] (\x-0.5,\y-0.5) -- (\x+0.5,\y-0.5) -- (\x+0.5,\y+0.5) -- (\x-0.5,\y+0.5) -- cycle;	
					
					\foreach \x in {8}	
					\foreach \y in {8}
					\draw [fill = MidnightBlue!30] (\x-0.5,\y-0.5) -- (\x+0.5,\y-0.5) -- (\x+0.5,\y+0.5) -- (\x-0.5,\y+0.5) -- cycle;	
					
					\foreach \x in {9}	
					\foreach \y in {8,...,16}
					\draw [fill = MidnightBlue!30] (\x-0.5,\y-0.5) -- (\x+0.5,\y-0.5) -- (\x+0.5,\y+0.5) -- (\x-0.5,\y+0.5) -- cycle;	
				\end{scope}

				\begin{scope}[yshift=-20cm]
					
					\foreach \x in {0,...,9}
					\draw (\x+0.5,0.5)--(\x+0.5,16.5);
					
					\foreach \y in {0,...,16}
					\draw (0.5,\y+0.5)--(9.5,\y+0.5);
					
					\node at (4,0) {\bf{1}};
					\node at (6,0) {\bf{0}};
					
					\node at (6,17.15) {\bf{0}};
					\node at (8,17.15) {\bf{1}};

					\foreach \x in {4}
					\foreach \y in {2,3,5,6,7,9,11,12,13,14}
					\draw (\x-0.25,\y-0.25) -- (\x+0.25,\y-0.25) -- (\x+0.25,\y+0.25) -- (\x-0.25,\y+0.25) -- cycle;
					
					\foreach \x in {6}
					\foreach \y in {4,5,6,7,13,14,15}
					\draw (\x-0.25,\y-0.25) -- (\x+0.25,\y-0.25) -- (\x+0.25,\y+0.25) -- (\x-0.25,\y+0.25) -- cycle;
					
					\draw [semithick,->,>=stealth] (4,0.5) -- (4,2+0.1);
					\draw [semithick,->,>=stealth] (6,0.5) -- (6,4+0.1);
					
					\draw [semithick,->,>=stealth] (4,7) -- (2,7) -- (2,13) -- (4+0.1,13);
					
					\draw [semithick,->,>=stealth] (4,9) -- (8,9) -- (8,16+0.75);
					\draw [semithick,->,>=stealth] (6,15) -- (6,16+0.75);

					\foreach \x in {1}	
					\foreach \y in {6,...,14}
					\draw [fill = MidnightBlue!30] (\x-0.5,\y-0.5) -- (\x+0.5,\y-0.5) -- (\x+0.5,\y+0.5) -- (\x-0.5,\y+0.5) -- cycle;	
					
					\foreach \x in {2}	
					\foreach \y in {6,14}
					\draw [fill = MidnightBlue!30] (\x-0.5,\y-0.5) -- (\x+0.5,\y-0.5) -- (\x+0.5,\y+0.5) -- (\x-0.5,\y+0.5) -- cycle;	
					
					\foreach \x in {3}	
					\foreach \y in {1,...,6,8,9,10,11,12,14,15}
					\draw [fill = MidnightBlue!30] (\x-0.5,\y-0.5) -- (\x+0.5,\y-0.5) -- (\x+0.5,\y+0.5) -- (\x-0.5,\y+0.5) -- cycle;	
					
					\foreach \x in {4}	
					\foreach \y in {8,15}
					\draw [fill = MidnightBlue!30] (\x-0.5,\y-0.5) -- (\x+0.5,\y-0.5) -- (\x+0.5,\y+0.5) -- (\x-0.5,\y+0.5) -- cycle;	
					
					\foreach \x in {5}	
					\foreach \y in {1,...,6,8,10,11,12,14,15,16}
					\draw [fill = MidnightBlue!30] (\x-0.5,\y-0.5) -- (\x+0.5,\y-0.5) -- (\x+0.5,\y+0.5) -- (\x-0.5,\y+0.5) -- cycle;	
					
					\foreach \x in {6}	
					\foreach \y in {8,10,11,12}
					\draw [fill = MidnightBlue!30] (\x-0.5,\y-0.5) -- (\x+0.5,\y-0.5) -- (\x+0.5,\y+0.5) -- (\x-0.5,\y+0.5) -- cycle;	
					
					\foreach \x in {7}	
					\foreach \y in {1,...,8,10,11,12,13,14,15,16}
					\draw [fill = MidnightBlue!30] (\x-0.5,\y-0.5) -- (\x+0.5,\y-0.5) -- (\x+0.5,\y+0.5) -- (\x-0.5,\y+0.5) -- cycle;	
					
					\foreach \x in {8}	
					\foreach \y in {8}
					\draw [fill = MidnightBlue!30] (\x-0.5,\y-0.5) -- (\x+0.5,\y-0.5) -- (\x+0.5,\y+0.5) -- (\x-0.5,\y+0.5) -- cycle;	
					
					\foreach \x in {9}	
					\foreach \y in {8,...,16}
					\draw [fill = MidnightBlue!30] (\x-0.5,\y-0.5) -- (\x+0.5,\y-0.5) -- (\x+0.5,\y+0.5) -- (\x-0.5,\y+0.5) -- cycle;	
				\end{scope}

				\begin{scope}[xshift=14cm,yshift=-20cm]
					
					\foreach \x in {0,...,9}
					\draw (\x+0.5,0.5)--(\x+0.5,16.5);
					
					\foreach \y in {0,...,16}
					\draw (0.5,\y+0.5)--(9.5,\y+0.5);
					
					\node at (4,0) {\bf{1}};
					\node at (6,0) {\bf{1}};
					
					\node at (6,17.15) {\bf{1}};
					\node at (8,17.15) {\bf{1}};

					\foreach \x in {4}
					\foreach \y in {2,3,5,6,7,9,11,12,13,14}
					\draw (\x-0.25,\y-0.25) -- (\x+0.25,\y-0.25) -- (\x+0.25,\y+0.25) -- (\x-0.25,\y+0.25) -- cycle;
					
					\foreach \x in {5}
					\foreach \y in {7,13}
					\draw (\x-0.25,\y-0.25) -- (\x+0.25,\y-0.25) -- (\x+0.25,\y+0.25) -- (\x-0.25,\y+0.25) -- cycle;
					
					\foreach \x in {6}
					\foreach \y in {4,5,7,13,15}
					\draw (\x-0.25,\y-0.25) -- (\x+0.25,\y-0.25) -- (\x+0.25,\y+0.25) -- (\x-0.25,\y+0.25) -- cycle;
					
					\draw [semithick,->,>=stealth] (4,0.5) -- (4,2+0.1);
					\draw [semithick,->,>=stealth] (6,0.5) -- (6,4+0.1);
					
					\draw [semithick,->,>=stealth] (4,7) -- (2,7) -- (2,13) -- (4+0.1,13);
					
					\draw [semithick,->,>=stealth] (4,9) -- (8,9) -- (8,16+0.75);
					\draw [semithick,->,>=stealth] (6,15) -- (6,16+0.75);

					\foreach \x in {1}	
					\foreach \y in {6,...,14}
					\draw [fill = MidnightBlue!30] (\x-0.5,\y-0.5) -- (\x+0.5,\y-0.5) -- (\x+0.5,\y+0.5) -- (\x-0.5,\y+0.5) -- cycle;	
					
					\foreach \x in {2}	
					\foreach \y in {6,14}
					\draw [fill = MidnightBlue!30] (\x-0.5,\y-0.5) -- (\x+0.5,\y-0.5) -- (\x+0.5,\y+0.5) -- (\x-0.5,\y+0.5) -- cycle;	
					
					\foreach \x in {3}	
					\foreach \y in {1,...,6,8,9,10,11,12,14,15}
					\draw [fill = MidnightBlue!30] (\x-0.5,\y-0.5) -- (\x+0.5,\y-0.5) -- (\x+0.5,\y+0.5) -- (\x-0.5,\y+0.5) -- cycle;	
					
					\foreach \x in {4}	
					\foreach \y in {8,15}
					\draw [fill = MidnightBlue!30] (\x-0.5,\y-0.5) -- (\x+0.5,\y-0.5) -- (\x+0.5,\y+0.5) -- (\x-0.5,\y+0.5) -- cycle;	
					
					\foreach \x in {5}	
					\foreach \y in {1,...,6,8,10,11,12,14,15,16}
					\draw [fill = MidnightBlue!30] (\x-0.5,\y-0.5) -- (\x+0.5,\y-0.5) -- (\x+0.5,\y+0.5) -- (\x-0.5,\y+0.5) -- cycle;	
					
					\foreach \x in {6}	
					\foreach \y in {8,10,11,12}
					\draw [fill = MidnightBlue!30] (\x-0.5,\y-0.5) -- (\x+0.5,\y-0.5) -- (\x+0.5,\y+0.5) -- (\x-0.5,\y+0.5) -- cycle;	
					
					\foreach \x in {7}	
					\foreach \y in {1,...,8,10,11,12,13,14,15,16}
					\draw [fill = MidnightBlue!30] (\x-0.5,\y-0.5) -- (\x+0.5,\y-0.5) -- (\x+0.5,\y+0.5) -- (\x-0.5,\y+0.5) -- cycle;	
					
					\foreach \x in {8}	
					\foreach \y in {8}
					\draw [fill = MidnightBlue!30] (\x-0.5,\y-0.5) -- (\x+0.5,\y-0.5) -- (\x+0.5,\y+0.5) -- (\x-0.5,\y+0.5) -- cycle;	
					
					\foreach \x in {9}	
					\foreach \y in {8,...,16}
					\draw [fill = MidnightBlue!30] (\x-0.5,\y-0.5) -- (\x+0.5,\y-0.5) -- (\x+0.5,\y+0.5) -- (\x-0.5,\y+0.5) -- cycle;	
				\end{scope}
				
			\end{tikzpicture}
			\caption {Four feasible solutions for a crossover gadget}
			\label {Fig_crossover_gadget_cases}
		\end{figure}
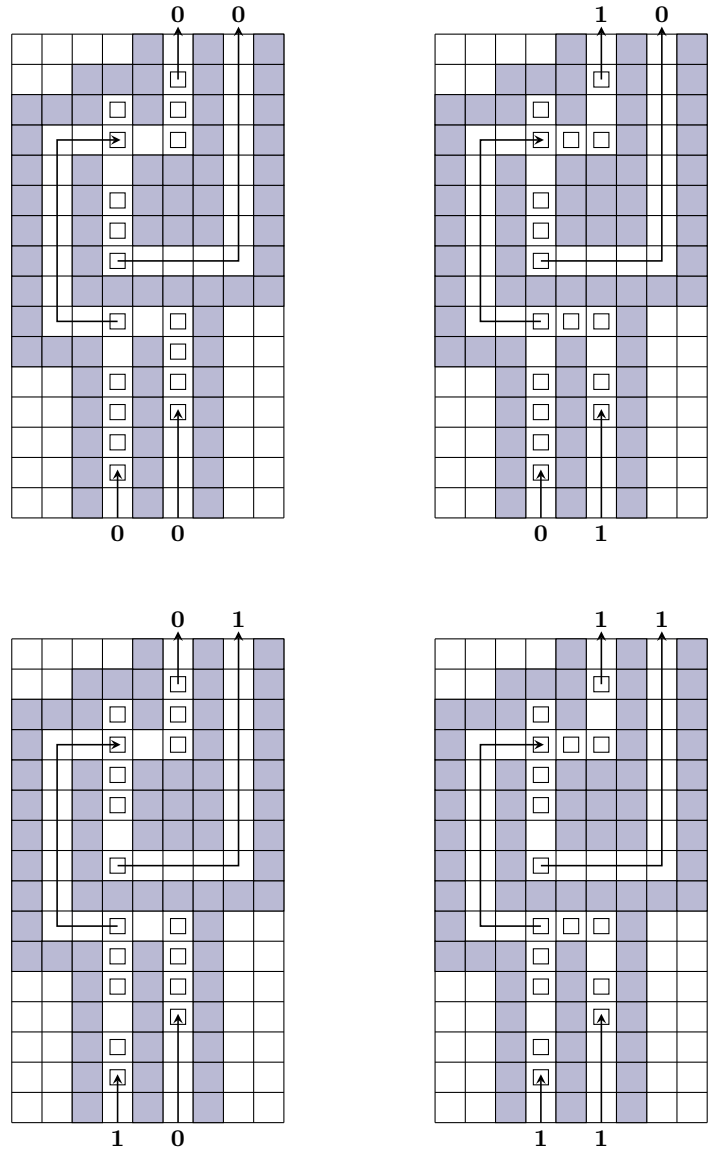

		\subsection*{Complexity of reduction and board size}
		
		Let's consider an instance of the \textsc{3-SAT} problem that consists of a collection of clauses $C = \{C_1,\ldots,C_m\}$ over Boolean variables $X=\{x_1,\ldots,x_n\}$. 
		
		Let's represent a CNF formula as a bipartite incidence graph with $n$ variables in one part and $m$ clauses in the other. An edge $(x,C)$ indicates that the variable $x$ appears in clause $C$. Hence, we obtain a bipartite graph with $n+m$ vertices and $3m$ edges, based on which we construct an \textsc{Evolomino} board. An example of an incidence graph of a CNF formula is shown in Fig.~\ref{Fig_full_board} with solid edges corresponding to non-inverted variables and dashed edges to inverted variables.
		
		Let's estimate the total number of gadgets required:
		\begin{itemize}
			\item $n$ variable and wire gadgets, one for each Boolean variable $x_1,\ldots,x_n$; 
			\item $m$ clause gadgets, one for each clause $C_1,\ldots,C_m$;
			\item $3m - n$ split gadgets, to assign a wire to each edge;
			\item at most $3m$ negation gadgets, one for each edge, for an inverted variable;
			\item at most $\binom{3m}{2}$ crossover gadgets, one for each crossing of edges.
		\end{itemize}
		
		Regarding the crossing number, we note that the incidence graph can always be drawn on the rectangular grid so that no three edges are crossing at a single grid cell, and any pair of edges crosses at most once. More precise estimates of the crossing number of a bipartite graph can be obtained from graph theory (see, for example, Turán's brick factory problem~\cite{Székely2016,Zarankiewicz1955}).
		
		Since each gadget, except for the wire, has a constant size, an instance of the \textsc{3-SAT} problem with $m$ clauses and $n$ variables can be reduced to an \textsc{Evolomino} puzzle on a board of size $O(m^2+n) \times O(m^2+n)$. As for the wires, they only need to transmit signals from the variable gadgets to the clause gadgets, so their length is proportional to the overall board size. 	Therefore, the reduction is performed in polynomial time.
		
		Note that both the board size and the complexity of the reduction can be significantly reduced if we exclude the crossover gadgets and consider the reduction from the NP-complete \textsc{planar 3-SAT} problem~\cite{Lichtenstein1982} where the incidence graph of a Boolean formula can be embedded on a plane, i.e. drawn without edge-crossings.  For related results on upper bounds for the area of rectangular grid drawings of planar graphs, see Rahman et al.~\cite{Rahman1998}.
		
		However, we decided that if the crossover is possible, it would be correct to consider a more general construction.

		\subsection*{An example of constructing the entire puzzle board}
		
		An example of constructing the \textsc{Evolomino} board for the \textsc{3-SAT} instance $C_1 \wedge C_2$, where $C_1 = (x_1 \vee \bar{x}_2 \vee x_3)$ and $C_2 = (x_2 \vee \bar{x}_3 \vee \bar{x}_4)$ is illustrated in Fig.~\ref{Fig_full_board}. 
		
		The board size is $42 \times 21$ and includes:
		\begin{itemize}
			\item 4 variable and wire gadgets for $x_1$, $x_2$, $x_3$, and $x_4$;
			\item 2 clause gadgets for $C_1$ and $C_2$;
			\item 2 split gadgets for $x_2$ and $x_3$;
			\item 3 negation gadgets for $(x_2,C_1)$, $(x_3,C_2)$, and $(x_4,C_2)$;
			\item 1 crossover gadget for $(x_2,C_2)$ and $(x_3,C_1)$.
		\end{itemize}
		
		An example of a puzzle solution that satisfies all the clauses and corresponds to the truth assignment $x_1 = 0$, $x_2 = 1$, $x_3 = 1$, $x_4 = 0$ is shown in Fig.~\ref{Fig_full_board_solution}.
		
		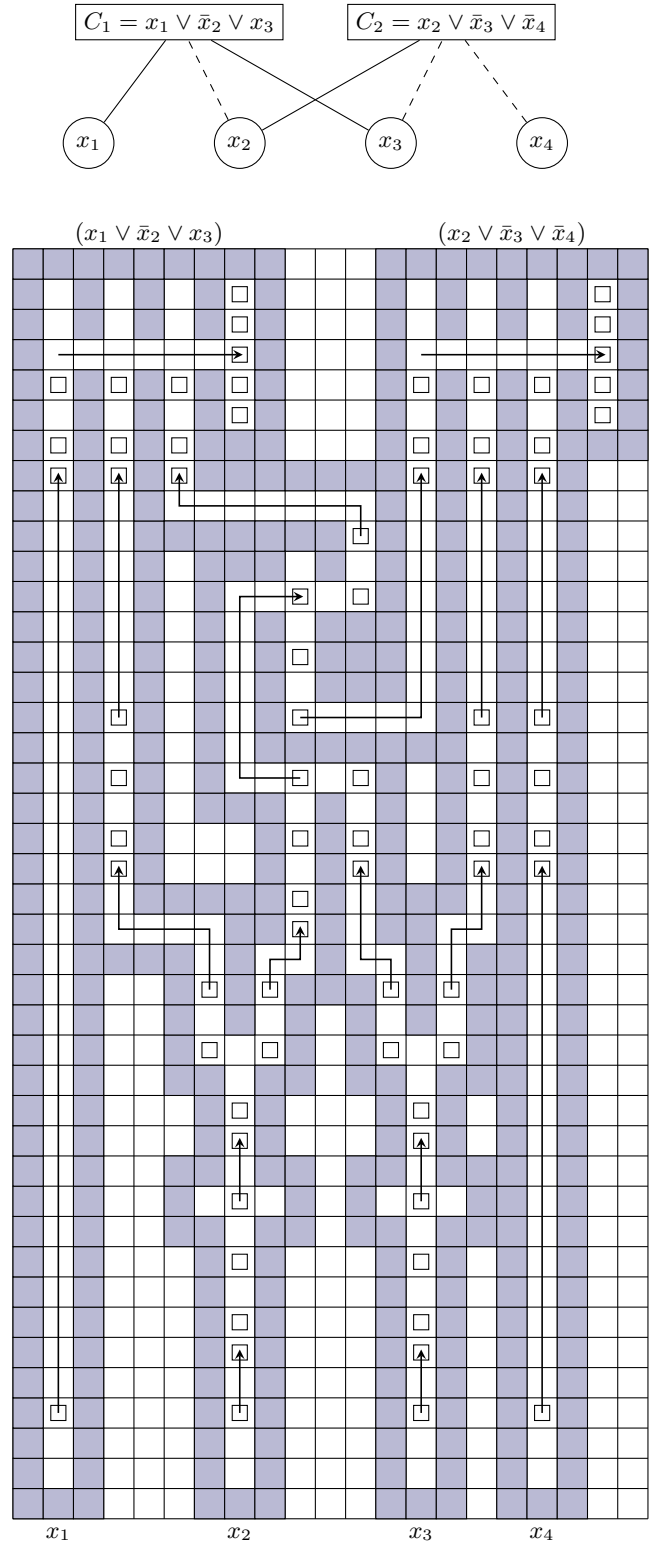
\begin{figure}[t]
			\centering
			\begin{tikzpicture}[scale=0.40]
				
				
				
				\foreach \x in {0,...,21}
				\draw (\x+0.5,0.5)--(\x+0.5,42.5);
				
				\foreach \x in {0,...,42}
				\draw (0.5,\x+0.5)--(21.5,\x+0.5);

				\foreach \x in {1}	
				\foreach \y in {1,...,42}
				\draw [fill = MidnightBlue!30] (\x-0.5,\y-0.5) -- (\x+0.5,\y-0.5) -- (\x+0.5,\y+0.5) -- (\x-0.5,\y+0.5) -- cycle;	
				
				\foreach \x in {2}	
				\foreach \y in {1,42}
				\draw [fill = MidnightBlue!30] (\x-0.5,\y-0.5) -- (\x+0.5,\y-0.5) -- (\x+0.5,\y+0.5) -- (\x-0.5,\y+0.5) -- cycle;	
				
				\foreach \x in {3}	
				\foreach \y in {1,...,38,40,41,42}
				\draw [fill = MidnightBlue!30] (\x-0.5,\y-0.5) -- (\x+0.5,\y-0.5) -- (\x+0.5,\y+0.5) -- (\x-0.5,\y+0.5) -- cycle;	
				
				\foreach \x in {4}	
				\foreach \y in {19,42}
				\draw [fill = MidnightBlue!30] (\x-0.5,\y-0.5) -- (\x+0.5,\y-0.5) -- (\x+0.5,\y+0.5) -- (\x-0.5,\y+0.5) -- cycle;	
				
				\foreach \x in {5}	
				\foreach \y in {19,21,22,...,38,40,41,42}
				\draw [fill = MidnightBlue!30] (\x-0.5,\y-0.5) -- (\x+0.5,\y-0.5) -- (\x+0.5,\y+0.5) -- (\x-0.5,\y+0.5) -- cycle;	
				
				\foreach \x in {6}	
				\foreach \y in {10,11,12,15,16,17,18,19,21,33,42}
				\draw [fill = MidnightBlue!30] (\x-0.5,\y-0.5) -- (\x+0.5,\y-0.5) -- (\x+0.5,\y+0.5) -- (\x-0.5,\y+0.5) -- cycle;	
				
				\foreach \x in {7}	
				\foreach \y in {1,...,10,12,13,14,15,21,24,25,...,33,35,36,37,38,40,41,42}
				\draw [fill = MidnightBlue!30] (\x-0.5,\y-0.5) -- (\x+0.5,\y-0.5) -- (\x+0.5,\y+0.5) -- (\x-0.5,\y+0.5) -- cycle;	
				
				\foreach \x in {8}	
				\foreach \y in {1,17,18,...,21,24,32,33,35,36,42}
				\draw [fill = MidnightBlue!30] (\x-0.5,\y-0.5) -- (\x+0.5,\y-0.5) -- (\x+0.5,\y+0.5) -- (\x-0.5,\y+0.5) -- cycle;	
				
				\foreach \x in {9}	
				\foreach \y in {1,...,10,12,13,14,15,20,21,22,23,24,26,27,28,29,30,32,33,35,36,...,42}
				\draw [fill = MidnightBlue!30] (\x-0.5,\y-0.5) -- (\x+0.5,\y-0.5) -- (\x+0.5,\y+0.5) -- (\x-0.5,\y+0.5) -- cycle;	
				
				\foreach \x in {10}	
				\foreach \y in {10,11,12,15,16,17,18,26,33,35}
				\draw [fill = MidnightBlue!30] (\x-0.5,\y-0.5) -- (\x+0.5,\y-0.5) -- (\x+0.5,\y+0.5) -- (\x-0.5,\y+0.5) -- cycle;	
				
				\foreach \x in {11}	
				\foreach \y in {18,...,24,26,28,29,30,32,33,35}
				\draw [fill = MidnightBlue!30] (\x-0.5,\y-0.5) -- (\x+0.5,\y-0.5) -- (\x+0.5,\y+0.5) -- (\x-0.5,\y+0.5) -- cycle;	
				
				\foreach \x in {12}	
				\foreach \y in {10,11,12,15,16,17,18,26,28,29,30,35}
				\draw [fill = MidnightBlue!30] (\x-0.5,\y-0.5) -- (\x+0.5,\y-0.5) -- (\x+0.5,\y+0.5) -- (\x-0.5,\y+0.5) -- cycle;	
				
				\foreach \x in {13}	
				\foreach \y in {1,...,10,12,13,14,15,20,21,...,26,28,29,...,42}
				\draw [fill = MidnightBlue!30] (\x-0.5,\y-0.5) -- (\x+0.5,\y-0.5) -- (\x+0.5,\y+0.5) -- (\x-0.5,\y+0.5) -- cycle;	
				
				\foreach \x in {14}	
				\foreach \y in {1,17,18,...,21,26,42}
				\draw [fill = MidnightBlue!30] (\x-0.5,\y-0.5) -- (\x+0.5,\y-0.5) -- (\x+0.5,\y+0.5) -- (\x-0.5,\y+0.5) -- cycle;	
				
				\foreach \x in {15}	
				\foreach \y in {1,...,10,12,13,14,15,21,22,...,38,40,41,42}
				\draw [fill = MidnightBlue!30] (\x-0.5,\y-0.5) -- (\x+0.5,\y-0.5) -- (\x+0.5,\y+0.5) -- (\x-0.5,\y+0.5) -- cycle;	
				
				\foreach \x in {16}	
				\foreach \y in {10,11,12,15,16,...,19,42}
				\draw [fill = MidnightBlue!30] (\x-0.5,\y-0.5) -- (\x+0.5,\y-0.5) -- (\x+0.5,\y+0.5) -- (\x-0.5,\y+0.5) -- cycle;	
				
				\foreach \x in {17}	
				\foreach \y in {1,...,38,40,41,42}
				\draw [fill = MidnightBlue!30] (\x-0.5,\y-0.5) -- (\x+0.5,\y-0.5) -- (\x+0.5,\y+0.5) -- (\x-0.5,\y+0.5) -- cycle;	
				
				\foreach \x in {18}	
				\foreach \y in {1,42}
				\draw [fill = MidnightBlue!30] (\x-0.5,\y-0.5) -- (\x+0.5,\y-0.5) -- (\x+0.5,\y+0.5) -- (\x-0.5,\y+0.5) -- cycle;	
				
				\foreach \x in {19}	
				\foreach \y in {1,...,38,40,41,42}
				\draw [fill = MidnightBlue!30] (\x-0.5,\y-0.5) -- (\x+0.5,\y-0.5) -- (\x+0.5,\y+0.5) -- (\x-0.5,\y+0.5) -- cycle;	
				
				\foreach \x in {20}	
				\foreach \y in {36,42}
				\draw [fill = MidnightBlue!30] (\x-0.5,\y-0.5) -- (\x+0.5,\y-0.5) -- (\x+0.5,\y+0.5) -- (\x-0.5,\y+0.5) -- cycle;	
				
				\foreach \x in {21}	
				\foreach \y in {36,...,42}
				\draw [fill = MidnightBlue!30] (\x-0.5,\y-0.5) -- (\x+0.5,\y-0.5) -- (\x+0.5,\y+0.5) -- (\x-0.5,\y+0.5) -- cycle;

				\foreach \x in {2}
				\foreach \y in {4,35,36,38}
				\draw (\x-0.25,\y-0.25) -- (\x+0.25,\y-0.25) -- (\x+0.25,\y+0.25) -- (\x-0.25,\y+0.25) -- cycle;
				
				\foreach \x in {4}
				\foreach \y in {22,23,25,27,35,36,38}
				\draw (\x-0.25,\y-0.25) -- (\x+0.25,\y-0.25) -- (\x+0.25,\y+0.25) -- (\x-0.25,\y+0.25) -- cycle;
				
				\foreach \x in {6}
				\foreach \y in {35,36,38}
				\draw (\x-0.25,\y-0.25) -- (\x+0.25,\y-0.25) -- (\x+0.25,\y+0.25) -- (\x-0.25,\y+0.25) -- cycle;
				
				\foreach \x in {7}
				\foreach \y in {16,18}
				\draw (\x-0.25,\y-0.25) -- (\x+0.25,\y-0.25) -- (\x+0.25,\y+0.25) -- (\x-0.25,\y+0.25) -- cycle;
				
				\foreach \x in {8}
				\foreach \y in {4,6,7,9,11,13,14,37,38,39,40,41}
				\draw (\x-0.25,\y-0.25) -- (\x+0.25,\y-0.25) -- (\x+0.25,\y+0.25) -- (\x-0.25,\y+0.25) -- cycle;
				
				\foreach \x in {9}
				\foreach \y in {16,18}
				\draw (\x-0.25,\y-0.25) -- (\x+0.25,\y-0.25) -- (\x+0.25,\y+0.25) -- (\x-0.25,\y+0.25) -- cycle;
				
				\foreach \x in {10}
				\foreach \y in {20,21,23,25,27,29,31}
				\draw (\x-0.25,\y-0.25) -- (\x+0.25,\y-0.25) -- (\x+0.25,\y+0.25) -- (\x-0.25,\y+0.25) -- cycle;
				
				\foreach \x in {12}
				\foreach \y in {22,23,25,31,33}
				\draw (\x-0.25,\y-0.25) -- (\x+0.25,\y-0.25) -- (\x+0.25,\y+0.25) -- (\x-0.25,\y+0.25) -- cycle;
				
				\foreach \x in {13}
				\foreach \y in {16,18}
				\draw (\x-0.25,\y-0.25) -- (\x+0.25,\y-0.25) -- (\x+0.25,\y+0.25) -- (\x-0.25,\y+0.25) -- cycle;
				
				\foreach \x in {14}
				\foreach \y in {4,6,7,9,11,13,14,35,36,38}
				\draw (\x-0.25,\y-0.25) -- (\x+0.25,\y-0.25) -- (\x+0.25,\y+0.25) -- (\x-0.25,\y+0.25) -- cycle;
				
				\foreach \x in {15}
				\foreach \y in {16,18}
				\draw (\x-0.25,\y-0.25) -- (\x+0.25,\y-0.25) -- (\x+0.25,\y+0.25) -- (\x-0.25,\y+0.25) -- cycle;
				
				\foreach \x in {16}
				\foreach \y in {22,23,25,27,35,36,38}
				\draw (\x-0.25,\y-0.25) -- (\x+0.25,\y-0.25) -- (\x+0.25,\y+0.25) -- (\x-0.25,\y+0.25) -- cycle;
				
				\foreach \x in {18}
				\foreach \y in {4,22,23,25,27,35,36,38}
				\draw (\x-0.25,\y-0.25) -- (\x+0.25,\y-0.25) -- (\x+0.25,\y+0.25) -- (\x-0.25,\y+0.25) -- cycle;
				
				\foreach \x in {20}
				\foreach \y in {37,38,39,40,41}
				\draw (\x-0.25,\y-0.25) -- (\x+0.25,\y-0.25) -- (\x+0.25,\y+0.25) -- (\x-0.25,\y+0.25) -- cycle;

				\foreach \x in {8,14}
				{
					\draw [semithick,->,>=stealth] (\x,4) -- (\x,6+0.1);
					\draw [semithick,->,>=stealth] (\x,11) -- (\x,13+0.1);			
				}

				\draw [semithick,->,>=stealth] (10,25) -- (8,25) -- (8,31) -- (10+0.1,31);
				
				\draw [semithick,->,>=stealth] (9,18) -- (9,19) -- (10,19) -- (10,20+0.1);
				\draw [semithick,->,>=stealth] (13,18) -- (13,19) -- (12,19) -- (12,22+0.1);
				
				\draw [semithick,->,>=stealth] (12,33) -- (12,34) -- (6,34) -- (6,35+0.1);
				
				\draw [semithick,->,>=stealth] (10,27) -- (14,27) -- (14,35+0.1);
				
				\draw [semithick,->,>=stealth] (7,18) -- (7,20) -- (4,20) -- (4,22+0.1);
				\draw [semithick,->,>=stealth] (15,18) -- (15,20) -- (16,20) -- (16,22+0.1);
				
				\foreach \x in {4,16,18}
				\draw [semithick,->,>=stealth] (\x,27) -- (\x,35+0.1);
				
				\draw [semithick,->,>=stealth] (2,4) -- (2,35+0.1);
				\draw [semithick,->,>=stealth] (18,4) -- (18,22+0.1);

				\draw [semithick,->,>=stealth] (2,39) -- (8+0.1,39);
				\draw [semithick,->,>=stealth] (14,39) -- (20+0.1,39);

				\node at (2,0) {$x_1$};
				\node at (8,0) {$x_2$};
				\node at (14,0) {$x_3$};
				\node at (18,0) {$x_4$};
				
				\node at (5,43) {$(x_1 \vee \bar{x}_2 \vee x_3)$};
				
				\node at (17,43) {$(x_2 \vee \bar{x}_3 \vee \bar{x}_4)$};

				\node (C1) [draw,rectangle] at (6,50) {$C_1=x_1 \vee \bar{x}_2 \vee x_3$};
				\node (C2) [draw,rectangle] at (15,50) {$C_2=x_2 \vee \bar{x}_3 \vee \bar{x}_4$};
				
				\node (x1) [draw,circle] at (3,46) {$x_1$};
				\node (x2) [draw,circle] at (8,46) {$x_2$};
				\node (x3) [draw,circle] at (13,46) {$x_3$};
				\node (x4) [draw,circle] at (18,46) {$x_4$};
				
				\draw (x1) -- (C1);
				\draw [dashed] (x2) -- (C1);
				\draw (x3) -- (C1);
				\draw (x2) -- (C2);
				\draw [dashed] (x3) -- (C2);
				\draw [dashed] (x4) -- (C2);
				
			\end{tikzpicture}
			\caption {An example of constructing an \textsc{Evolomino} puzzle for the CNF formula\\ $(x_1 \vee \bar{x}_2 \vee x_3) \wedge (x_2 \vee \bar{x}_3 \vee \bar{x}_4)$}
			\label {Fig_full_board}
		\end{figure}

		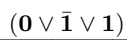
\begin{figure}[t]
			\centering
			\begin{tikzpicture}[scale=0.40]
				
				
				
				\foreach \x in {0,...,21}
				\draw (\x+0.5,0.5)--(\x+0.5,42.5);
				
				\foreach \x in {0,...,42}
				\draw (0.5,\x+0.5)--(21.5,\x+0.5);

				\foreach \x in {1}	
				\foreach \y in {1,...,42}
				\draw [fill = MidnightBlue!30] (\x-0.5,\y-0.5) -- (\x+0.5,\y-0.5) -- (\x+0.5,\y+0.5) -- (\x-0.5,\y+0.5) -- cycle;	
				
				\foreach \x in {2}	
				\foreach \y in {1,42}
				\draw [fill = MidnightBlue!30] (\x-0.5,\y-0.5) -- (\x+0.5,\y-0.5) -- (\x+0.5,\y+0.5) -- (\x-0.5,\y+0.5) -- cycle;	
				
				\foreach \x in {3}	
				\foreach \y in {1,...,38,40,41,42}
				\draw [fill = MidnightBlue!30] (\x-0.5,\y-0.5) -- (\x+0.5,\y-0.5) -- (\x+0.5,\y+0.5) -- (\x-0.5,\y+0.5) -- cycle;	
				
				\foreach \x in {4}	
				\foreach \y in {19,42}
				\draw [fill = MidnightBlue!30] (\x-0.5,\y-0.5) -- (\x+0.5,\y-0.5) -- (\x+0.5,\y+0.5) -- (\x-0.5,\y+0.5) -- cycle;	
				
				\foreach \x in {5}	
				\foreach \y in {19,21,22,...,38,40,41,42}
				\draw [fill = MidnightBlue!30] (\x-0.5,\y-0.5) -- (\x+0.5,\y-0.5) -- (\x+0.5,\y+0.5) -- (\x-0.5,\y+0.5) -- cycle;	
				
				\foreach \x in {6}	
				\foreach \y in {10,11,12,15,16,17,18,19,21,33,42}
				\draw [fill = MidnightBlue!30] (\x-0.5,\y-0.5) -- (\x+0.5,\y-0.5) -- (\x+0.5,\y+0.5) -- (\x-0.5,\y+0.5) -- cycle;	
				
				\foreach \x in {7}	
				\foreach \y in {1,...,10,12,13,14,15,21,24,25,...,33,35,36,37,38,40,41,42}
				\draw [fill = MidnightBlue!30] (\x-0.5,\y-0.5) -- (\x+0.5,\y-0.5) -- (\x+0.5,\y+0.5) -- (\x-0.5,\y+0.5) -- cycle;	
				
				\foreach \x in {8}	
				\foreach \y in {1,17,18,...,21,24,32,33,35,36,42}
				\draw [fill = MidnightBlue!30] (\x-0.5,\y-0.5) -- (\x+0.5,\y-0.5) -- (\x+0.5,\y+0.5) -- (\x-0.5,\y+0.5) -- cycle;	
				
				\foreach \x in {9}	
				\foreach \y in {1,...,10,12,13,14,15,20,21,22,23,24,26,27,28,29,30,32,33,35,36,...,42}
				\draw [fill = MidnightBlue!30] (\x-0.5,\y-0.5) -- (\x+0.5,\y-0.5) -- (\x+0.5,\y+0.5) -- (\x-0.5,\y+0.5) -- cycle;	
				
				\foreach \x in {10}	
				\foreach \y in {10,11,12,15,16,17,18,26,33,35}
				\draw [fill = MidnightBlue!30] (\x-0.5,\y-0.5) -- (\x+0.5,\y-0.5) -- (\x+0.5,\y+0.5) -- (\x-0.5,\y+0.5) -- cycle;	
				
				\foreach \x in {11}	
				\foreach \y in {18,...,24,26,28,29,30,32,33,35}
				\draw [fill = MidnightBlue!30] (\x-0.5,\y-0.5) -- (\x+0.5,\y-0.5) -- (\x+0.5,\y+0.5) -- (\x-0.5,\y+0.5) -- cycle;	
				
				\foreach \x in {12}	
				\foreach \y in {10,11,12,15,16,17,18,26,28,29,30,35}
				\draw [fill = MidnightBlue!30] (\x-0.5,\y-0.5) -- (\x+0.5,\y-0.5) -- (\x+0.5,\y+0.5) -- (\x-0.5,\y+0.5) -- cycle;	
				
				\foreach \x in {13}	
				\foreach \y in {1,...,10,12,13,14,15,20,21,...,26,28,29,...,42}
				\draw [fill = MidnightBlue!30] (\x-0.5,\y-0.5) -- (\x+0.5,\y-0.5) -- (\x+0.5,\y+0.5) -- (\x-0.5,\y+0.5) -- cycle;	
				
				\foreach \x in {14}	
				\foreach \y in {1,17,18,...,21,26,42}
				\draw [fill = MidnightBlue!30] (\x-0.5,\y-0.5) -- (\x+0.5,\y-0.5) -- (\x+0.5,\y+0.5) -- (\x-0.5,\y+0.5) -- cycle;	
				
				\foreach \x in {15}	
				\foreach \y in {1,...,10,12,13,14,15,21,22,...,38,40,41,42}
				\draw [fill = MidnightBlue!30] (\x-0.5,\y-0.5) -- (\x+0.5,\y-0.5) -- (\x+0.5,\y+0.5) -- (\x-0.5,\y+0.5) -- cycle;	
				
				\foreach \x in {16}	
				\foreach \y in {10,11,12,15,16,...,19,42}
				\draw [fill = MidnightBlue!30] (\x-0.5,\y-0.5) -- (\x+0.5,\y-0.5) -- (\x+0.5,\y+0.5) -- (\x-0.5,\y+0.5) -- cycle;	
				
				\foreach \x in {17}	
				\foreach \y in {1,...,38,40,41,42}
				\draw [fill = MidnightBlue!30] (\x-0.5,\y-0.5) -- (\x+0.5,\y-0.5) -- (\x+0.5,\y+0.5) -- (\x-0.5,\y+0.5) -- cycle;	
				
				\foreach \x in {18}	
				\foreach \y in {1,42}
				\draw [fill = MidnightBlue!30] (\x-0.5,\y-0.5) -- (\x+0.5,\y-0.5) -- (\x+0.5,\y+0.5) -- (\x-0.5,\y+0.5) -- cycle;	
				
				\foreach \x in {19}	
				\foreach \y in {1,...,38,40,41,42}
				\draw [fill = MidnightBlue!30] (\x-0.5,\y-0.5) -- (\x+0.5,\y-0.5) -- (\x+0.5,\y+0.5) -- (\x-0.5,\y+0.5) -- cycle;	
				
				\foreach \x in {20}	
				\foreach \y in {36,42}
				\draw [fill = MidnightBlue!30] (\x-0.5,\y-0.5) -- (\x+0.5,\y-0.5) -- (\x+0.5,\y+0.5) -- (\x-0.5,\y+0.5) -- cycle;	
				
				\foreach \x in {21}	
				\foreach \y in {36,...,42}
				\draw [fill = MidnightBlue!30] (\x-0.5,\y-0.5) -- (\x+0.5,\y-0.5) -- (\x+0.5,\y+0.5) -- (\x-0.5,\y+0.5) -- cycle;

				\foreach \x in {2}
				\foreach \y in {2,3,4,35,36,37,38}
				\draw (\x-0.25,\y-0.25) -- (\x+0.25,\y-0.25) -- (\x+0.25,\y+0.25) -- (\x-0.25,\y+0.25) -- cycle;
				
				\foreach \x in {4}
				\foreach \y in {22,23,25,26,27,35,36,37,38}
				\draw (\x-0.25,\y-0.25) -- (\x+0.25,\y-0.25) -- (\x+0.25,\y+0.25) -- (\x-0.25,\y+0.25) -- cycle;
				
				\foreach \x in {6}
				\foreach \y in {35,36,38,39,40,41}
				\draw (\x-0.25,\y-0.25) -- (\x+0.25,\y-0.25) -- (\x+0.25,\y+0.25) -- (\x-0.25,\y+0.25) -- cycle;
				
				\foreach \x in {7}
				\foreach \y in {11,16,18}
				\draw (\x-0.25,\y-0.25) -- (\x+0.25,\y-0.25) -- (\x+0.25,\y+0.25) -- (\x-0.25,\y+0.25) -- cycle;
				
				\foreach \x in {8}
				\foreach \y in {4,6,7,9,10,11,13,14,15,16,37,38,39,40,41}
				\draw (\x-0.25,\y-0.25) -- (\x+0.25,\y-0.25) -- (\x+0.25,\y+0.25) -- (\x-0.25,\y+0.25) -- cycle;
				
				\foreach \x in {9}
				\foreach \y in {11,16,18}
				\draw (\x-0.25,\y-0.25) -- (\x+0.25,\y-0.25) -- (\x+0.25,\y+0.25) -- (\x-0.25,\y+0.25) -- cycle;
				
				\foreach \x in {10}
				\foreach \y in {20,21,23,24,25,27,29,30,31,32}
				\draw (\x-0.25,\y-0.25) -- (\x+0.25,\y-0.25) -- (\x+0.25,\y+0.25) -- (\x-0.25,\y+0.25) -- cycle;
				
				\foreach \x in {11}
				\foreach \y in {25,31}
				\draw (\x-0.25,\y-0.25) -- (\x+0.25,\y-0.25) -- (\x+0.25,\y+0.25) -- (\x-0.25,\y+0.25) -- cycle;
				
				\foreach \x in {12}
				\foreach \y in {22,23,25,31,33}
				\draw (\x-0.25,\y-0.25) -- (\x+0.25,\y-0.25) -- (\x+0.25,\y+0.25) -- (\x-0.25,\y+0.25) -- cycle;
				
				\foreach \x in {13}
				\foreach \y in {11,16,18}
				\draw (\x-0.25,\y-0.25) -- (\x+0.25,\y-0.25) -- (\x+0.25,\y+0.25) -- (\x-0.25,\y+0.25) -- cycle;
				
				\foreach \x in {14}
				\foreach \y in {4,6,7,9,10,11,13,14,15,16,35,36,38,39,40}
				\draw (\x-0.25,\y-0.25) -- (\x+0.25,\y-0.25) -- (\x+0.25,\y+0.25) -- (\x-0.25,\y+0.25) -- cycle;
				
				\foreach \x in {15}
				\foreach \y in {11,16,18}
				\draw (\x-0.25,\y-0.25) -- (\x+0.25,\y-0.25) -- (\x+0.25,\y+0.25) -- (\x-0.25,\y+0.25) -- cycle;
				
				\foreach \x in {16}
				\foreach \y in {22,23,25,26,27,35,36,37,38}
				\draw (\x-0.25,\y-0.25) -- (\x+0.25,\y-0.25) -- (\x+0.25,\y+0.25) -- (\x-0.25,\y+0.25) -- cycle;
				
				\foreach \x in {18}
				\foreach \y in {2,3,4,22,23,24,25,27,35,36,38,39,40,41}
				\draw (\x-0.25,\y-0.25) -- (\x+0.25,\y-0.25) -- (\x+0.25,\y+0.25) -- (\x-0.25,\y+0.25) -- cycle;
				
				\foreach \x in {20}
				\foreach \y in {37,38,39,40,41}
				\draw (\x-0.25,\y-0.25) -- (\x+0.25,\y-0.25) -- (\x+0.25,\y+0.25) -- (\x-0.25,\y+0.25) -- cycle;

				\foreach \x in {8,14}
				{
					\draw [semithick,->,>=stealth] (\x,4) -- (\x,6+0.1);
					\draw [semithick,->,>=stealth] (\x,11) -- (\x,13+0.1);			
				}

				\draw [semithick,->,>=stealth] (10,25) -- (8,25) -- (8,31) -- (10+0.1,31);
				
				\draw [semithick,->,>=stealth] (9,18) -- (9,19) -- (10,19) -- (10,20+0.1);
				\draw [semithick,->,>=stealth] (13,18) -- (13,19) -- (12,19) -- (12,22+0.1);
				
				\draw [semithick,->,>=stealth] (12,33) -- (12,34) -- (6,34) -- (6,35+0.1);
				
				\draw [semithick,->,>=stealth] (10,27) -- (14,27) -- (14,35+0.1);
				
				\draw [semithick,->,>=stealth] (7,18) -- (7,20) -- (4,20) -- (4,22+0.1);
				\draw [semithick,->,>=stealth] (15,18) -- (15,20) -- (16,20) -- (16,22+0.1);
				
				\foreach \x in {4,16,18}
				\draw [semithick,->,>=stealth] (\x,27) -- (\x,35+0.1);
				
				\draw [semithick,->,>=stealth] (2,4) -- (2,35+0.1);
				\draw [semithick,->,>=stealth] (18,4) -- (18,22+0.1);

				\draw [semithick,->,>=stealth] (2,39) -- (8+0.1,39);
				\draw [semithick,->,>=stealth] (14,39) -- (20+0.1,39);

				\node at (2,0) {\bf{0}};
				\node at (8,0) {\bf{1}};
				\node at (14,0) {\bf{1}};
				\node at (18,0) {\bf{0}};
				
				\node at (5,43) {$(\mathbf{0} \vee \bar{\mathbf{1}} \vee \mathbf{1})$};
				\node at (17,43) {$(\mathbf{1} \vee \bar{\mathbf{1}} \vee \bar{\mathbf{0}})$};

			\end{tikzpicture}
			\caption {Solution to the puzzle corresponding to the CNF formula $(x_1 \vee \bar{x}_2 \vee x_3)$\\ $\wedge (x_2 \vee \bar{x}_3 \vee \bar{x}_4)$ under the truth assignment $x_1 = 0$, $x_2=1$, $x_3=1$, $x_4=0$}
			\label {Fig_full_board_solution}
		\end{figure}

		\subsection*{Completing the proof}
		
		It remains to note that all gadgets, except for the clause gadgets, always admit a valid solution, while the clause gadget is solvable if and only if at least one of the literals is true.
		Thus, the puzzle has a solution if and only if the corresponding instance of the \textsc{3-SAT} problem is satisfiable.
		This makes the \textsc{Evolomino} problem NP-hard, and since, by Lemma~\ref{Lemma_Evolomino_in_NP}, it also belongs to the class NP, we conclude that \textsc{Evolomino} is NP-complete.	\qed
	\end{proof}
	
	Moreover, once we fix a truth assignment in an instance of \textsc{3-SAT}, the filling pattern of the resulting instance of \textsc{Evolomino} is uniquely determined. Thus, our reduction is parsimonious, i.e., it establishes a bijection between the solution sets of the two problems. Consequently, the number of satisfying truth assignments to the original CNF formula equals the number of solutions to the resulting \textsc{Evolomino} puzzle. Since the \textsc{Counting 3-SAT} is $\texttt{\#}$P-complete~\cite{Valiant1979}, we obtain the following corollary.  
	
	\begin{corollary}
		\textsc{Counting Evolomino} is $\texttt{\#}$P-complete.
	\end{corollary}
	
	Let us recall that the complexity class $\texttt{\#}$P, introduced by Valiant~\cite{Valiant1979}, contains the counting versions of NP decision problems. For more information on $\texttt{\#}$P-complete problems and parsimonious reduction, see the book by Goldreich~\cite{Goldreich2008}.

	\section{Conclusion}
	
	In this paper, we have proved, via a reduction from \textsc{3-SAT}, that the \textsc{Evolomino} puzzle is NP-complete. Furthermore, since our reduction is parsimonious and preserves the number of distinct solutions, the corresponding counting problem, \textsc{Counting Evolomino}, is also $\texttt{\#}$P-complete.
	
	Evolomino is a relatively new puzzle for which little is known from a theoretical computer science perspective. Several directions for future research appear promising. On the one hand, the development of algorithms for the general NP-hard case of \textsc{Evolomino}, such as integer linear programming models or backtracking approaches.
	On the other hand, the NP-completeness result implies only that some instances, in particular those that correspond to the \textsc{3-SAT} problem, are computationally intractable (assuming P $\neq$ NP). It would be worthwhile to investigate additional constraints or special cases under which subproblems of \textsc{Evolomino} become polynomially solvable.
	
	\begin{credits}
		\subsubsection{\ackname} 
		The research is supported by the P.G. Demidov Yaroslavl State University Project VIP-016.
		
		\subsubsection{\discintname}
		The author has no competing interests to declare that are relevant to the content of this article.
	\end{credits}
	%
	%
	%
	\bibliographystyle{splncs04}
	\bibliography{Nikolaev_arXiv_2025}

\end{document}